\documentclass[11pt,a4paper]{article}
\usepackage[total={6.50in, 9.5in}]{geometry}
\usepackage{amsmath,graphicx,xcolor,url,titling}
\usepackage[round]{natbib}

\usepackage{amsthm}
\theoremstyle{plain}

\newtheorem{lemma}{Lemma}
\theoremstyle{remark} 
\title{\textbf{Development of robust $X$-bar charts with unequal sample sizes}}
\begin{small}
\author{
\textsf{Chanseok Park} \\
Applied Statistics Laboratory\\
Department of Industrial Engineering\\
Pusan National University\\
Busan 46241, Korea
\and
\textsf{Linhan Ouyang} \\
College of Economics and Management\\
Nanjing University of Aeronautics and Astronautics\\
Nanjing, Jiangsu 211106, China
\and
\textsf{Min Wang}\thanks{Corresponding author.
Email: \texttt{min.wang3@utsa.edu}}\\
Department of Management Science and Statistics \\
The University of Texas at San Antonio \\
San Antonio, TX, USA
}
\end{small}
\date{}

\begin{document}
\maketitle

\begin{abstract}
The traditional variable control charts, such as the $\bar{X}$ chart, are 
widely used to monitor variation in a process. 
They have been shown to perform well for monitoring processes 
under the general assumptions that the observations are normally distributed without data contamination and that the sample sizes from the process are all equal. However, these two assumptions may not be met and satisfied in many practical applications and thus make them potentially limited for widespread application 
especially in production processes. 
In this paper, we alleviate this limitation by providing a novel method 
for constructing the robust $\bar{X}$ control charts, which can simultaneously deal with both data contamination 
and unequal sample sizes. 
The proposed method for the process parameters
is optimal in a sense of the best linear unbiased estimation. 
Numerical results  from extensive Monte Carlo simulations and a real data analysis reveal that traditional control charts seriously underperform for monitoring process in the presence of data contamination and are extremely sensitive to even a single contaminated value, while the proposed robust control charts outperform in a manner that is comparable with the traditional ones, whereas they are far superior when the data are contaminated by outliers.

\medskip
\textbf{Keywords:}
Average run length; BLUE estimator; relative efficiency; robustness; statistical process control. 
\end{abstract}

\section{Introduction}

Variables control charts provide a powerful tool for monitoring variation in a process where the measurement is a variable \citep{Montgomery:2013a}, which can be measured on a continuous scale, such as length, pressure, width, temperature, and volume, in a time-ordered sequence. Since their introduction by \cite{Shewhart:1926b}, variable control charts, such as the Shewhart control chart, have been extensively adopted to detect if a process is in a state of control in different fields including industrial manufacturing, service sectors, and healthcare systems, to name just a few.  The statistical control charts can generally be classified into two phases: Phase-I and Phase-II \citep{Vining:2009,Montgomery:2013a}. First, in Phase-I monitoring, we establish reliable control limits with a clean set of process data and in Phase-II analysis, we use the control limits obtained in Phase-I in the monitoring of a process by comparing the statistic for each successive subgroup as future observations are obtained. 
We refer the interested reader to the collection of textbooks and papers in 
\cite{Montgomery:2013a, Montgomery:2019, Fara:Sani:Mont:2019, Wood:Falt:2019} for details.

In most Phase II monitoring situations, the majority of variable control charts  usually depend on the unknown process parameters that need to be estimated  based on an in-control Phase I sample or historical data, 
indicating the quality of these data plays an important role in determining the performance of the Phase-II control charts.  It deserves mentioning that the constructions of traditional control charts and other related techniques for statistical process monitoring usually requires the following two general assumptions: (i) the observations are normally distributed without data contamination \citep{Wheel:2010, Khakifirooz:2021} and (ii) the sample sizes from the process are all equal. However, these two assumptions are quite restrictive for developing efficient control charts for monitoring the manufacturing process.   

The traditional control charts, such as the $\bar{X}$ chart, are constructed using the sample mean and standard deviation which are sensitive to data contamination, because their breakdown points are zero. Recently, \cite{Yao:Chak:2021} pointed out the importance of accounting for the effect of parameter estimation in the Shewhart-type Phase I chart, since the direct use of the control charts with inaccurate parameter estimates for process monitoring could increase the rate of false alarms to an unacceptable level and even result in misleading results in practical applications, especially when the Phase I data are contaminated by outliers due to the measurement errors, the volatile operating conditions, among others. To overcome this limitation, researchers developed robust control charts by incorporating robust estimators for the location and scale parameters. For instance, \cite{Rocke:1989} obtained the control limits based on the trimmed mean and the interquartile range. 
\cite{Allo:Ragh:1991} and \cite{Papp:Ben:1996} used the Hodges-Lehmann control chart 
for location. \cite{Janacek/Meikle:1997} used the median, and  \cite{Abu-Shawiesh:2008} considered the median absolute deviation (MAD) for constructing robust control charts.

In addition, the construction of traditional control charts often requires the sample sizes to be all equal in the process, whereas this requirement may not be met in practice due to missing observations, cost constraints, etc. 
As commented by \cite{Kim/Reynolds:2005}, since the importance of each variable is different, it is necessary to take the variable sample size issue into consideration in the field of statistical process control. To tackle the issue of unbalanced sample sizes, \cite{Kosz:2018} proposed a risk-based concept for the design of an $\bar X$ chart with variable sample size, in which the optimal sample size is determined by the genetic algorithm and the Nelder-Mead direct algorithm. 

However, to our knowledge, no study has proposed 
the construction of robust $\bar{X}$ chart which can simultaneously deal with both data contamination and 
unequal sample sizes. 
To fill this gap, we first mitigate the effects of data contamination for the parameter estimation 
by incorporating robust estimators with high breakdown point values, 
such as the median and Hodges-Lehmann (HL) estimator for the location parameter and the MAD and Shamos estimators for the scale estimation parameter. To overcome the issue of unequal sample sizes, we take a weighted average approach in calculating $\bar X$ based on the sample sizes, since estimating a parameter from the larger samples is more reliable than that from smaller samples and the weighted average should be used instead of the simple average \citep{Burr:1969}. 
As an illustration from the location estimation, we assume that there are $m$ samples (subgroups) and $n_i$ is the sample size of the $i$th sample from a manufacturing process. Then the simple average, $\bar{\bar{X}}_A = \sum_{i=1}^m \bar{X}_i/m$,
is inferior to the weighted average, $\bar{\bar{X}}_B = \sum_{i=1}^m n_i \bar{X}_i / \sum_{i=1}^m n_i$, 
in a sense that $\mathrm{Var}(\bar{\bar{X}}_A) \ge \mathrm{Var}(\bar{\bar{X}}_B)$.
Note that $\bar{\bar{X}}_A$ has an unweighted average of the estimates while
$\bar{\bar{X}}_B$ has a weighted average with weights proportional to the sample sizes.
This idea pays off with the sample means. 
However, for other estimators such as the median and MAD estimators, 
although it looks appealing to use the weights proportional to the sample sizes,
the weights simply proportional to the sample sizes 
can make the estimation worse as will be shown in Section~\ref{SEC:process}.

The remainder of this paper is organized as follows. 
In Section~\ref{SEC:process}, 
we proposed the unbiased location and scale estimators for the process parameters with unequal sample sizes. 
In Section~\ref{SEC:performance}, we carry out simulation studies 
to investigate the performance of the proposed estimation methods under no data contamination 
and under data contamination. 
In Section~\ref{section4}, we briefly introduce how to construct robust control charts with the proposed robust estimators and compare the performance of the robust control charts under consideration. An illustrative example is provided 
in Section~\ref{section5} for demonstrative purposes. 
Finally, several concluding remarks are given in Section~\ref{section6}. 

\section{Process parameter estimation with unequal sample sizes\label{SEC:process}}
In this section, we propose the unbiased location and scale estimators 
for the process parameters under the assumption that each sample has different sample sizes and that the underlying distribution is normally distributed.
Let $X_{ij}$ be the $i$th sample (subgroup) of
size $n_i$ from a Phase-I process for $j=1,2,\ldots,n_i$ and
$i=1,2,\ldots,m$ .
We assume that $X_{ij}$'s are independent and identically distributed (iid) random variables from a normal distribution with location $\mu$ and variance $\sigma^2$.

\subsection{Estimation of the location parameter}\label{SEC:location}
The following location estimators for the population mean parameter are widely used 
in the quality control literature, 
\begin{align}
\bar{\bar{X}}_A &= \frac{\bar{X}_1 + \bar{X}_2 +\cdots+ \bar{X}_m}{m}
                = \frac{1}{m} \sum_{i=1}^m \bar{X}_i            \label{EQ:XA}
\intertext{and}
\bar{\bar{X}}_B &= \frac{n_1\bar{X}_1+n_2\bar{X}_2+\cdots+ n_m\bar{X}_m}{n_1+n_2+\ldots+n_m}
                = \frac{1}{N} \sum_{i=1}^m n_i \bar{X}_i,       \label{EQ:XB}
\end{align}
where $\bar{X}_i = \sum_{j=1}^{n_i} X_{ij}/n_i$ and $N=\sum_{i=1}^m n_i$; see, for example, Equations (6.2) and (6.30) of \cite{Montgomery:2013a}
for $\bar{\bar{X}}_A$ and $\bar{\bar{X}}_B$, respectively.
Their variances are given by
$\mathrm{Var}(\bar{\bar{X}}_A) = {\sigma^2} \sum_{i=1}^{m} n_i^{-1} / {m^2}$
and
$\mathrm{Var}(\bar{\bar{X}}_B) = \sigma^2/N$.
Using the inequality of the arithmetic and  harmonic means,
it can be easily shown that $\mathrm{Var}(\bar{\bar{X}}_A) \ge \mathrm{Var}(\bar{\bar{X}}_B)$
\citep{Park/Wang:2020a}. 
In this case, the estimator $\bar{\bar{X}}_B$ is actually optimal in a sense that
it is the best linear unbiased estimator (BLUE), which will be shown later.

Analogous to $\bar{\bar{X}}_A$ in (\ref{EQ:XA}) and $\bar{\bar{X}}_B$ in (\ref{EQ:XB}),
we can generalize the above with the following location estimators to estimate the population mean parameter $\mu$. 
\begin{align}
\bar{\hat{\mu}}_A &= \frac{\hat{\mu}_1 + \hat{\mu}_2 +\cdots+ \hat{\mu}_m}{m}
                = \frac{1}{m} \sum_{i=1}^m \hat{\mu}_i            \label{EQ:muA}
\intertext{and}
\bar{\hat{\mu}}_B &= \frac{n_1\hat{\mu}_1+n_2\hat{\mu}_2+\cdots+ n_m\hat{\mu}_m}{n_1+n_2+\ldots+n_m}
                = \frac{1}{N} \sum_{i=1}^m n_i \hat{\mu}_i,       \label{EQ:muB}
\end{align}
where $\hat{\mu}_i$ denotes the unbiased location estimator of $\mu$ with the $i$th sample.
It can be easily shown that they are unbiased, such that 
$E(\bar{\hat{\mu}}_A)=\mu$ and $E(\bar{\hat{\mu}}_B)=\mu$ and that their variances are given by 
$\mathrm{Var}(\bar{\hat{\mu}}_A) = \sum_{i=1}^{m} V_i / {m^2}$
and 
$\mathrm{Var}(\bar{\hat{\mu}}_B) = \sum_{i=1}^{m} n_i^2 V_i / {N^2}$, where
$V_i=\mathrm{Var}(\hat{\mu}_i)$.

When the sample mean is considered, it is always true that 
$\mathrm{Var}(\bar{\bar{X}}_A) \ge \mathrm{Var}(\bar{\bar{X}}_B)$ as afore-mentioned.
However, when other estimators are considered, $\mathrm{Var}(\bar{\hat{\mu}}_A) \ge \mathrm{Var}(\bar{\hat{\mu}}_B)$ does not hold in general.
For example, when the sample median is used with $n_1=4$ and $n_2=5$, by using numerical results from \cite{Park/Wang:2022b} and \cite{Park/Kim/Wang:2022}, we obtain 
$V_1=\mathrm{Var}(\hat{\mu}_1) = (1.1930/4)\sigma^2 = 0.29825\sigma^2$ and 
$V_2=\mathrm{Var}(\hat{\mu}_2) = (1.4339/5)\sigma^2 = 0.28678\sigma^2$, indicating that $\mathrm{Var}(\bar{\hat{\mu}}_A) = 0.146\sigma^2 < \mathrm{Var}(\bar{\hat{\mu}}_B) = 0.147\sigma^2$. In what follows, we propose the BLUE for the location parameter $\mu$. 
\begin{lemma} \label{LEM:location}
The BLUE for the location parameter $\mu$ is given by 
\begin{equation}
\bar{\hat{\mu}}_C = \frac{\sum_{i=1}^{m} (\hat{\mu}_i / \nu_i^2)}{\sum_{i=1}^{m} (1/\nu_i^2)},
\label{EQ:muC}
\end{equation}
where $\nu_i^2$ is the variance of $\hat{\mu}_i$ under the standard normal distribution.
\end{lemma}
\begin{proof}
We consider a linear estimator in the form of $\bar{\hat{\mu}}_C = \sum_{i=1}^m w_i \hat{\mu}_i$
where $E(\hat{\mu}_i)=\mu$ as aforementioned. 
Since $\bar{\hat{\mu}}_C$ is unbiased such that $E(\bar{\hat{\mu}}_C) = \mu$, we have $\sum_{i=1}^m w_i=1$. Thus, our objective is to minimize $\mathrm{Var}(\bar{\hat{\mu}}_C)=\sum_{i=1}^m w_i^2 V_i$ 
with the constraint $\sum_{i=1}^m w_i=1$. We can set up the auxiliary function with the Lagrange multiplier $\lambda$ given by 
\[
\Psi = \sum_{i=1}^m w_i^2 V_i - \lambda \Big(\sum_{i=1}^m w_i-1 \Big).
\]
Differentiating $\Psi$ with respect to $w_i$ and setting it to zero, 
we have $2w_i V_i - \lambda=0$, which results in 
$w_i = \lambda/(2V_i)$. Since $\sum_{i=1}^m w_i=1$, we have $\sum_{i=1}^m \lambda/(2V_i)=1$
so that $\lambda=2/\sum_{i=1}^m(1/V_i)$. 
Then we obtain $w_i=(1/V_i)/\sum_{i=1}^m(1/V_i)$.
Since the normal distribution is a part of the location-scale family
of distributions, we have 
$V_i = \sigma^2 \nu_i^2$. Thus, we have 
$w_i=(1/\nu_i^2)/\sum_{i=1}^m(1/\nu_i^2)$, which completes the proof.
\end{proof}

We can easily show that 
$\mathrm{Var}(\bar{\hat{\mu}}_A) \ge \mathrm{Var}(\bar{\hat{\mu}}_C)$ using
 the inequality of the arithmetic and  harmonic means
and that $\mathrm{Var}(\bar{\hat{\mu}}_B) \ge \mathrm{Var}(\bar{\hat{\mu}}_C)$
using the Cauchy-Schwarz inequality, where $\mathrm{Var}(\bar{\hat{\mu}}_C)=1/\sum_{i=1}^m(1/V_i)$.
As mentioned above, there is no clear inequality relation between $\mathrm{Var}(\bar{\hat{\mu}}_A)$
and $\mathrm{Var}(\bar{\hat{\mu}}_B)$, unless the sample mean is considered.
For the special case of the sample mean $\hat{\mu}_i=\bar{X}_i$, the above $\bar{\hat{\mu}}_C$
is the same as $\bar{\hat{\mu}}_B$, since $V_i = \mathrm{Var}(\hat{\mu}_i) = \sigma^2/n_i$, 
which implies that $\bar{\bar{X}}_B = \sum_{i=1}^m n_i \bar{X}_i / N$ is optimal.

We here develop a robust BLUE location estimator using the median and HL estimators. 
Let $X_{ij} = \mu + \sigma Z_{ij}$, where $Z_{ij}$ are 
iid with $N(0,1)$. 
Let $\mathrm{median}(X_i) = \mathrm{median}\{X_{i1}, \ldots, X_{in_i}\}$.
We have $\mathrm{median}(X_i) = \mu + \sigma\cdot \mathrm{median}(Z_i)$, 
so that $\mathrm{Var}\{\mathrm{median}(X_i)\} = \sigma^2 \mathrm{Var}\{\mathrm{median}(Z_i)\}$. 
The HL estimator is calculated by the median of all pairwise (Walsh) averages of the observations,
which is given by 
\[
\mathrm{HL}(X_i) = \mathop{\mathrm{median}} \Big( \frac{X_{ik}+X_{i\ell}}{2} \Big).
\]
Here the median of all Walsh averages can be calculated for the three cases 
\citep{Park/Kim/Wang:2022}: 
(i) $k<\ell$, (ii) $k\le\ell$, and  (iii) $\forall(k,\ell)$, for $k,\ell = 1,2,\ldots,n_i$. Thus, we can have three versions as below
\begin{align*}
\mathrm{HL}1(X_i) &= \mathop{\mathrm{median}}_{k<\ell} \Big( \frac{X_{ik}+X_{i\ell}}{2} \Big),  \\
\mathrm{HL}2(X_i) &= \mathop{\mathrm{median}}_{k\le \ell} \Big( \frac{X_{ik}+X_{i\ell}}{2} \Big),   \\
\intertext{and} 
\mathrm{HL}3(X_i) &= \mathop{\mathrm{median}}_{\forall(k,\ell)} \Big( \frac{X_{ik}+X_{i\ell}}{2} \Big).
\end{align*}
It is noteworthy that all three versions are asymptotically equivalent \citep{Serfling:2011}. 
In this paper, we use $\mathrm{HL}1$, denoted as $\mathrm{HL}$ for brevity.
For the HL case, we also have 
\[
\mathrm{HL}(X_i) = \mathop{\mathrm{median}}_{k<\ell} \Big( \frac{X_{ik}+X_{i\ell}}{2} \Big)
            = \mu + \sigma \cdot \mathop{\mathrm{median}}_{k<\ell} \Big( \frac{Z_{ik}+Z_{i\ell}}{2} \Big) 
            = \mu + \sigma \cdot \mathrm{HL}(Z_i).
\]
Thus, we have 
\[
\mathrm{Var}\left\{ \mathrm{HL}(X_i) \right\} = \sigma^2\cdot \mathrm{Var}\left\{ \mathrm{HL}(Z_i)  \right\} .
\]
By using the empirical variances of the sample median and HL estimator from \cite{Park/Wang:2022b} and \cite{Park/Kim/Wang:2022} along with Lemma~\ref{LEM:location},  we can easily obtain the robust BLUE estimators for the population mean based on the sample median and the HL estimator.

\subsection{Estimation of the scale parameter}\label{SEC:scale}
\cite{Park/Wang:2020a} proposed the following unbiased estimators for the population scale $\sigma$, 
which are given by 
\begin{align} 
\bar{S}_A&=  \frac{S_1/c_4(n_1) + S_2/c_4(n_2) +\cdots+ S_m/c_4(n_m)}{m}
          = \frac{1}{m} \sum_{i=1}^m \frac{S_i}{c_4(n_i)}, \label{EQ:SA} 
\intertext{and}
\bar{S}_B&= \frac{S_1+S_2+\cdots+ S_m}{c_4(n_1)+c_4(n_2)+\cdots+ c_4(n_m)}
          = \frac{\sum_{i=1}^m S_i}{\sum_{i=1}^m c_4(n_i)}, \label{EQ:SB} 
\end{align}
where ${S}_i^2  = \sum_{j=1}^{n_i} (X_{ij}-\bar{X}_i)^2 / (n_i-1)$
and $c_4(n_i) = \sqrt{{2}/{(n_i-1)}} \cdot {\Gamma(n_i/2)}/{\Gamma( (n_i-1)/2 )}$.
In a similar way as done above, we can consider the following estimators
for the scale parameter given by 
\begin{align} 
\bar{\hat{\sigma}}_A&=  \frac{\hat{\sigma}_1/C_1 + \hat{\sigma}_2/C_2 +\cdots+ \hat{\sigma}_m/C_m}{m}
          = \frac{1}{m} \sum_{i=1}^m \frac{\hat{\sigma}_i}{C_i}, \label{EQ:sigmaA} \\
\intertext{and}
\bar{\hat{\sigma}}_B&= \frac{\hat{\sigma}_1+\hat{\sigma}_2+\cdots+\hat{\sigma}_m}{C_1+C_2+\cdots+C_m}
          = \frac{\sum_{i=1}^m \hat{\sigma}_i}{\sum_{i=1}^m C_i}, \label{EQ:sigmaB} 
\end{align}
where $\hat{\sigma}_i$ denotes the estimator of $\sigma$ with the $i$the sample
and $C_i = E(\hat{\sigma}_i)/\sigma$.
Then both $\bar{\hat{\sigma}}_A$ and $\bar{\hat{\sigma}}_B$ are unbiased. 
The variances of $\bar{\hat{\sigma}}_A$ and $\bar{\hat{\sigma}}_B$ are given by 
\[
\mathrm{Var}(\bar{\hat{\sigma}}_A) = \frac{1}{m^2} \sum_{i=1}^m \frac{V_i}{C_i^2}
\textrm{~~and~~}
\mathrm{Var}(\bar{\hat{\sigma}}_B) = \frac{\sum_{i=1}^m V_i}{ (\sum_{i=1}^m C_i)^2},
\]
where $V_i = \mathrm{Var}(\hat{\sigma}_i)$. We observe from \cite{Park/Wang:2020a} that $\mathrm{Var}(\bar{S}_A) \ge \mathrm{Var}(\bar{S}_B)$
using the Chebyshev's sum inequality along with the property that $c_4(x)$ is increasing and $1/c_4(x)^2-1$ is decreasing. 

However, the inequality of  $\mathrm{Var}(\bar{\hat{\sigma}}_A) \ge \mathrm{Var}(\bar{\hat{\sigma}}_B)$ does not hold in general.  As an illustration, we consider the MAD given by  
\[
\mathrm{MAD}(X_i) 
= \frac{\displaystyle{\mathop\mathrm{median}_{1\le j\le n_i}}|X_{ij}-\tilde{\mu}_i|}{\Phi^{-1}({3}/{4})},
\]
where $\tilde{\mu}_i = \mathrm{median}(X_i)$.
Here $\Phi^{-1}({3}/{4})$ is needed to make this estimator Fisher-consistent 
\citep{Fisher:1922} for the standard deviation $\sigma$ under the normal distribution. 
Since this MAD is not unbiased with a finite sample size, \cite{Park/Wang:2022b} and \cite{Park/Kim/Wang:2022} obtained the empirical unbiasing factor, denoted by $c_5(n_i)$, for this MAD. 
Thus, we can obtain the \emph{unbiased} MAD with a finite sample size
 for the standard deviation $\sigma$ under the normal distribution
which is 
\begin{equation}  \label{EQ:unibaseMAD}
\frac{\mathrm{MAD}(X_i)}{c_5(n_i)}.
\end{equation}
For example, when the unbiased MAD is used with the sample sizes $n_1=4$ and $n_2=5$, 
we have $\mathrm{Var}(\bar{\hat{\sigma}}_A) = 0.167\sigma^2$ and 
$\mathrm{Var}(\bar{\hat{\sigma}}_B) = 0.168\sigma^2$ so that  we have 
$\mathrm{Var}(\bar{\hat{\sigma}}_A) < \mathrm{Var}(\bar{\hat{\sigma}}_B)$ in this case.

\begin{figure}[tp]
\centering\includegraphics{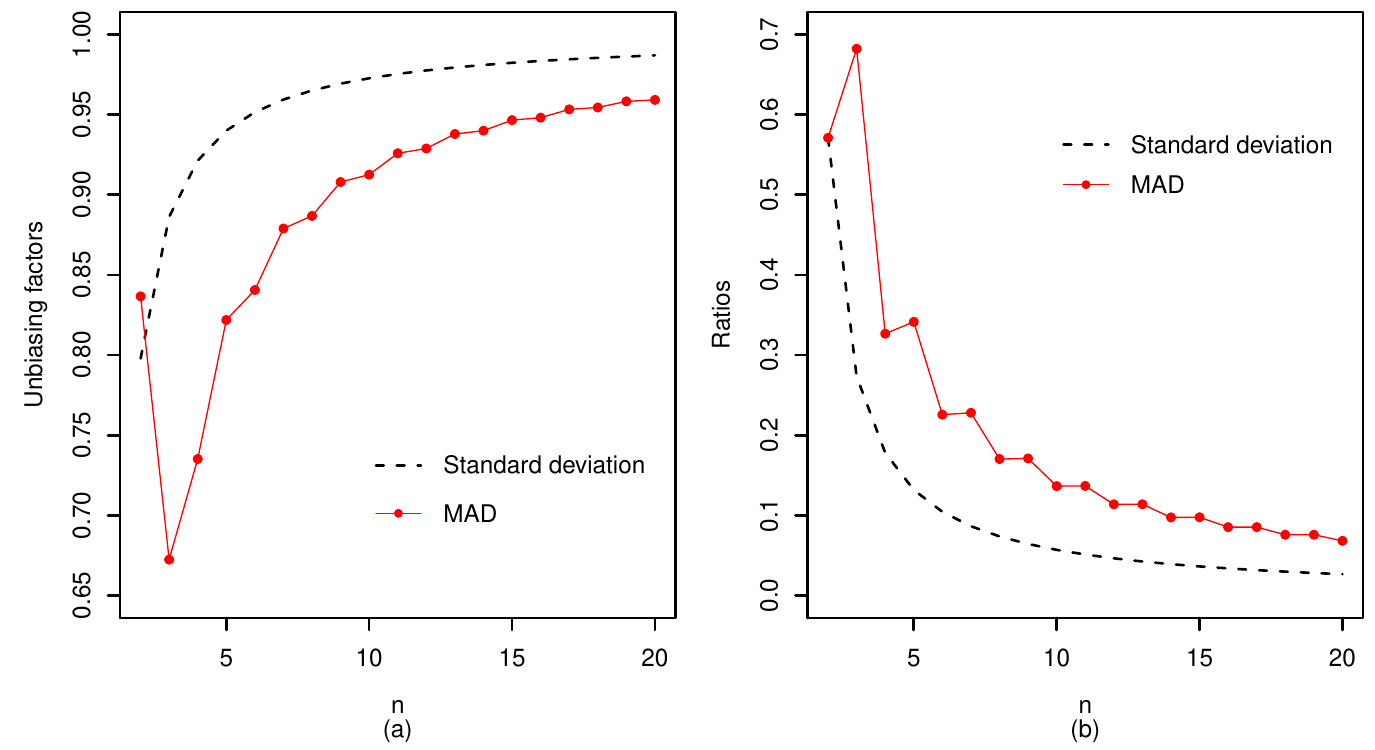}  
\caption{The unbiasing factors and ratios of the squared unbiasing factor to the variance.
\label{FIG:FactorsMAD}}
\end{figure}

The term $c_4(x)$ is the unbiasing factor for the standard deviation
and $1-c_4(x)^2$ is the variance of the standard deviation under the standard normal distribution.
Then we have $\mathrm{Var}(S/c_4(n)) = \sigma^2(1-c_4(x)^2)/c_4(x)^2$.
Thus, $1/c_4(x)^2-1=(1-c_4(x)^2)/c_4(x)^2$ can be thought of 
as the ratio of the variance to the squared unbiasing factor for the estimator under the standard normal distribution.
In Figure~\ref{FIG:FactorsMAD}, we draw the unbiasing factors and 
the ratios of the variance to the squared unbiasing factor for the standard deviation and the unbiased MAD
under the standard normal distribution.
This figure indicates that the inequality of 
$\mathrm{Var}(\bar{\hat{\sigma}}_A) \ge \mathrm{Var}(\bar{\hat{\sigma}}_B)$ may not hold for the unbiased MAD.
In what follows, we propose the BLUE for the scale parameter.

\begin{lemma} \label{LEM:scale}
The BLUE for the scale parameter is given by
\begin{equation}
\bar{\hat{\sigma}}_C 
= \frac{\sum_{i=1}^{m} (\gamma_i/\tau_i^2) \hat{\sigma}_i}{\sum_{i=1}^{m} (\gamma_i^2/\tau_i^2)},
\label{EQ:sigmaC}
\end{equation}
where $\gamma_i$ and $\tau_i^2$ are the expectation and
variance of $\hat{\sigma}_i$ under the standard normal distribution, respectively.
\end{lemma}
\begin{proof}
We consider a linear estimator in the form of $\bar{\hat{\sigma}}_C = \sum_{i=1}^m w_i \hat{\sigma}_i$.
Since $E(\bar{\hat{\sigma}}_C) = \sigma$, we have $\sum_{i=1}^m w_i C_i=1$, 
where $C_i = E(\hat{\sigma}_i)/\sigma$ as mentioned above.
With the constraint $\sum_{i=1}^m w_i C_i=1$, 
the estimator $\bar{\hat{\sigma}}_C$ is guaranteed to be unbiased.
Thus, our objective is to minimize $\mathrm{Var}(\bar{\hat{\sigma}}_C)=\sum_{i=1}^m w_i^2 V_i$
with this constraint, where $V_i = \mathrm{Var}(\hat{\sigma}_i)$.
We set up the auxiliary function with the Lagrange multiplier $\lambda$ given by
\[
\Psi = \sum_{i=1}^m w_i^2 V_i - \lambda \Big(\sum_{i=1}^m w_i C_i - 1 \Big).
\]
Differentiating $\Psi$ with respect to $w_i$ and setting it to zero,
we have $2w_i V_i - \lambda C_i=0$, which results in
$w_i = \lambda C_i/(2V_i)$. Since $\sum_{i=1}^m w_i C_i = 1$, 
we have $\sum_{i=1}^m \lambda C_i^2/(2V_i)=1$
so that $\lambda=2/\sum_{i=1}^m(C_i^2/V_i)$.
Thus, we have $w_i = (C_i/V_i)/\sum_{i=1}^m(C_i^2/V_i)$.
Since the normal distribution is a part of the location-scale family, 
we have $C_i=\gamma_i$ and $V_i = \sigma^2 \tau_i^2$, so that we have 
$w_i=(\gamma_i/\tau_i^2)/\sum_{i=1}^m(\gamma_i^2/\tau_i^2)$, which completes the proof.
\end{proof}

The term $\gamma_i$ in Lemma~\ref{LEM:scale} can be regarded as an unbiasing factor for $\hat{\sigma}_i$.
Thus, if $\hat{\sigma}_i$ is an unbiased estimator of $\sigma$, then we have $\gamma_i=1$. Thus, we can rewrite (\ref{EQ:sigmaC}) as 
\begin{equation} \label{EQ:unbiasedscale}
\bar{\hat{\sigma}}_C 
= \frac{\sum_{i=1}^{m} (1/\tau_i^2) \hat{\sigma}_i}{\sum_{i=1}^{m} (1/\tau_i^2)}.
\end{equation}
When $X_{ij} = \mu + \sigma Z_{ij}$, we have $\mathrm{MAD}(X_i) = \sigma \mathrm{MAD}(Z_i)$. 
We can then incorporate Lemma~\ref{LEM:scale} into $\mathrm{MAD}(X_i)$.
We also consider another robust scale estimator called Shamos estimator  which is given by 
\[
\mathrm{Shamos}(X_i)
= \frac{\displaystyle\mathop{\mathrm{median}}_{k<\ell} \big( |X_{ik}-X_{i\ell}| \big)}{\sqrt{2}\,\Phi^{-1}(3/4)},
\]
which is also biased with a finite sample size. 
Here $\sqrt{2}\,\Phi^{-1}(3/4)$ is needed to make this estimator Fisher-consistent 
for the standard deviation $\sigma$ under the normal distribution \citep{Levy/etc:2011}.
To make it unbiased with a finite sample, we adopt the empirical unbiasing factor $c_6(n_i)$ obtained by \cite{Park/Wang:2022b} and \cite{Park/Kim/Wang:2022} and obtain the \emph{unbiased} Shamos estimator given by 
\begin{equation} \label{EQ:unbiasedShamos}
\frac{\mathrm{Shamos}(X_i)}{c_6(n_i)}. 
\end{equation}
Thus, we can easily obtain BLUE for the Shamos case  using (\ref{EQ:unbiasedShamos}) along with (\ref{EQ:unbiasedscale}). Analogous to the case of the MAD, we have $\mathrm{Shamos}(X_i) = \sigma \mathrm{Shamos}(Z_i)$,
so that we can incorporate Lemma~\ref{LEM:scale} into $\mathrm{Shamos}(X_i)$.

It deserves mentioning that 
\cite{Burr:1969} suggested the unbiased estimator of $\sigma$,
\begin{equation} \label{EQ:SD}
\bar{S}_D = \frac{S_p}{c_4(N-m+1)},
\end{equation}
where $S_p^2 = {\sum_{i=1}^m (n_i-1) S_i^2 }/{(N-m)}$. 
As shown in Theorem~4 of \cite{Park/Wang:2020a}, $\mathrm{Var}(\bar{S}_C) \ge \mathrm{Var}(\bar{S}_D)$.
Let $\bar{\hat{\sigma}}_D$ be the estimator for the scale using the MAD with analogy 
to $\bar{S}_D$ above. We can consider 
\[
\mathrm{MAD}(X) 
= \frac{\displaystyle{\mathop\mathrm{median}_{1\le i\le m, 1\le j\le n_i}}|X_{ij}-\tilde{\mu}|}{\Phi^{-1}({3}/{4})}, 
\]
where $\tilde{\mu} = \displaystyle\mathop\mathrm{median}_{1\le i\le m, 1\le j\le n_i}(X_{ij})$.
Then the unbiased estimator of $\sigma$ is given by 
\[
\bar{\hat{\sigma}}_D = \frac{\mathrm{MAD}(X)}{c_5(N)} .
\]
The terms related to the variance and squared unbiasing factor 
(possibly their ratio) under the standard normal play a role in the calculation of its variance. 
We observe from Figure~\ref{FIG:FactorsMAD} that these values change like 
a saw due to the nature of the different median calculations with even and odd sample sizes. 
Thus, as with the case of $\bar{\hat{\sigma}}_A$ and $\bar{\hat{\sigma}}_B$,
the inequality of $\mathrm{Var}(\bar{\hat{\sigma}}_C) \ge \mathrm{Var}(\bar{\hat{\sigma}}_D)$ 
may not hold. \cite{Park/Wang:2020a} numerically showed that the performances of $\bar{S}_C$ and $\bar{S}_D$ 
are very close in all the performance measures although $\bar{S}_D$ is slightly better with respect to the relative efficiency (RE), the average run length (ARL), 
and the standard deviation of run length (SRDL).
Thus, we here consider only the three types of estimators denoted by $A$, $B$, $C$ as above.

\section{Performance of the process parameter estimators\label{SEC:performance}}
In this section, we compare the performance of the proposed methods in (\ref{EQ:muC}) and  (\ref{EQ:sigmaC})
under no data contamination and under data contamination as well. 

\subsection{Performance under no contamination}
First, we use the RE defined as the ratio of their variances. Let $\hat{\theta}_0$ and $\hat{\theta}_1$ be two estimators with $\hat{\theta}_0$ being a baseline or reference
estimator. Then the RE (in percentage) of $\hat{\theta}_1$ with respect to  $\hat{\theta}_0$ is defined as 
\begin{equation} \label{EQ:RE}
\mathrm{RE}(\hat{\theta}_1 \mid \hat{\theta}_0) 
  = \frac{\mathrm{Var}(\hat{\theta}_0)}{\mathrm{Var}(\hat{\theta}_1)} \times 100 .
\end{equation}
We consider three samples with the four different scenarios. 
The sample sizes of each scenario are given by 
\begin{itemize}
\item[(a)] $n_1=3$, $n_2=10$, $n_3=17$, 
\item[(b)] $n_1=5$, $n_2=10$, $n_3=15$, 
\item[(c)] $n_1=7$, $n_2=10$, $n_3=13$,
\item[(d)] $n_1=9$, $n_2=10$, $n_3=11$.
\end{itemize}

We generate $X_{ij}$'s from the normal distribution with mean $\mu_0=100$ and standard deviation $\sigma_0=10$. For each of the above four scenarios, 
we estimate $\mu_0$ using the three types of estimation methods
($\bar{\hat{\mu}}_A$, $\bar{\hat{\mu}}_B$, $\bar{\hat{\mu}}_C$) with the mean, median, and HL estimators, respectively. 
Similarly, we estimate $\sigma_0$
using the three types of estimation methods 
($\bar{\hat{\sigma}}_A$, $\bar{\hat{\sigma}}_B$, $\bar{\hat{\sigma}}_C$)
with the standard deviation, MAD, and Shamos estimators, respectively.
We repeat each simulation one million times ($I=10^6$) to obtain the empirical variances and RE values
which are provided in Table~\ref{TBL:varRE}. 
We here calculate the RE values with respect to the BLUE of the mean for $\mu_0$
and the BLUE of the standard deviation for $\sigma_0$.
Thus, the RE of the mean with type C is always 100\% and that of the standard deviation with type C is also 
100\%.
We summarize the values of RE in Figure~\ref{FIG:RElocationpure} (location estimators) 
and Figure~\ref{FIG:REscalepure} (scale estimators).

As seen in Figures \ref{FIG:RElocationpure} and \ref{FIG:REscalepure},
the mean and standard deviation with type C is the best as expected. 
For the location case, the median does not perform well, but the HL with type C is nearly as efficient as 
the mean. For the scale case, the MAD performs poorly, but the Shamos with type C performs well. 
It should be noted that for most cases, the performance with type B is better than that with type A, but
for the Shamos case with small samples, the performance with type A is better than that with type B
as seen in Figure \ref{FIG:REscalepure} (a) and (b).

\begin{table}[tbp]
\renewcommand\arraystretch{0.8}
\caption{\label{TBL:varRE}Estimated variances and relative efficiency.}
\begin{center}
\begin{small}
\begin{tabular}{cl*{5}{r}l*{3}{r}}
\hline
 && \multicolumn{5}{c}{Location estimator} && \multicolumn{3}{c}{Scale estimator} \\
 && mean & median & $\mathrm{HL}1$ & $\mathrm{HL}2$ & $\mathrm{HL}3$ && SD & MAD & Shamos \\
\hline
         && \multicolumn{9}{c}{Sample sizes: $n_1=3$,\quad $n_2=10$,\quad $n_3=17$} \\ 
Variance && \\
A        && 5.4645 & 7.5319 & 5.8887 & 5.6722 & 5.8952 && 4.0281 & 10.0563 & 4.6174  \\
B        && 3.3335 & 4.8851 & 3.5346 & 3.5459 & 3.5465 && 3.7081 &  7.7644 & 5.5640  \\
C        && 3.3335 & 4.8711 & 3.5344 & 3.5452 & 3.5463 && 1.9040 &  4.8842 & 2.3609  \\
RE (\%)  &&  \\
A        &&  61.00 &  44.26 &  56.61 &  58.77 &  56.55 &&  47.27 &  18.93 &  41.24  \\
B        &&  100.0 &  68.24 &  94.31 &  94.01 &  93.99 &&  51.35 &  24.52 &  34.22  \\
C        &&  100.0 &  68.44 &  94.32 &  94.03 &  94.00 &&  100.0 &  38.98 &  80.65  \\
\hline
         && \multicolumn{9}{c}{Sample sizes: $n_1=5$,\quad $n_2=10$,\quad $n_3=15$} \\ 
Variance && \\
A        && 4.0707 & 5.8565 & 4.3195 & 4.3723 & 4.3577 && 2.5087 & 6.4081 & 3.3051  \\
B        && 3.3335 & 4.8830 & 3.5306 & 3.5681 & 3.5493 && 2.4585 & 6.0195 & 3.4176  \\
C        && 3.3335 & 4.8732 & 3.5306 & 3.5681 & 3.5493 && 1.9070 & 4.8877 & 2.4273  \\
RE (\%)  && \\
A        &&  81.89 &  56.92 &  77.17 &  76.24 &  76.49 &&  76.01 &  29.76 &  57.70  \\
B        && 100.00 &  68.27 &  94.41 &  93.42 &  93.91 &&  76.57 &  31.68 &  55.80  \\
C        && 100.00 &  68.40 &  94.41 &  93.42 &  93.92 && 100.00 &  39.01 &  78.56  \\    
\hline
         && \multicolumn{9}{c}{Sample sizes: $n_1=7$,\quad $n_2=10$,\quad $n_3=13$} \\ 
Variance &&  \\
A        && 3.5502 & 5.1726 & 3.7679 & 3.8218 & 3.7973 && 2.0747 & 5.3268 & 2.6692  \\
B        && 3.3335 & 4.8812 & 3.5334 & 3.5796 & 3.5561 && 2.0644 & 5.2467 & 2.6906  \\
C        && 3.3335 & 4.8734 & 3.5334 & 3.5795 & 3.5561 && 1.9082 & 4.8935 & 2.4319  \\
RE (\%)  && \\
A        &&  93.90 &  64.45 &  88.47 &  87.22 &  87.79 &&  91.97 &  35.82 &  71.49  \\
B        && 100.00 &  68.29 &  94.34 &  93.13 &  93.74 &&  92.43 &  36.37 &  70.92  \\
C        && 100.00 &  68.40 &  94.34 &  93.13 &  93.74 && 100.00 &  38.99 &  78.46  \\
\hline
         && \multicolumn{9}{c}{Sample sizes: $n_1=9$,\quad $n_2=10$,\quad $n_3=11$} \\ 
Variance && \\
A        && 3.3552 & 4.9101 & 3.5543 & 3.5995 & 3.5734 && 1.9249 & 4.9478 & 2.4608 \\
B        && 3.3335 & 4.8812 & 3.5318 & 3.5751 & 3.5498 && 1.9240 & 4.9419 & 2.4629 \\
C        && 3.3335 & 4.8737 & 3.5318 & 3.5751 & 3.5499 && 1.9084 & 4.8914 & 2.4403 \\
RE (\%)  && \\
A        &&  99.35 &  67.89 &  93.79 &  92.61 &  93.29 &&  99.14 &  38.57 & 77.56   \\
B        && 100.00 &  68.29 &  94.39 &  93.24 &  93.91 &&  99.18 &  38.62 & 77.49   \\
C        && 100.00 &  68.40 &  94.39 &  93.24 &  93.91 && 100.00 &  39.01 & 78.20   \\
\hline
\end{tabular}
\end{small}
\end{center}
\end{table}

\bigskip
\begin{figure}[thbp]
\centering\includegraphics{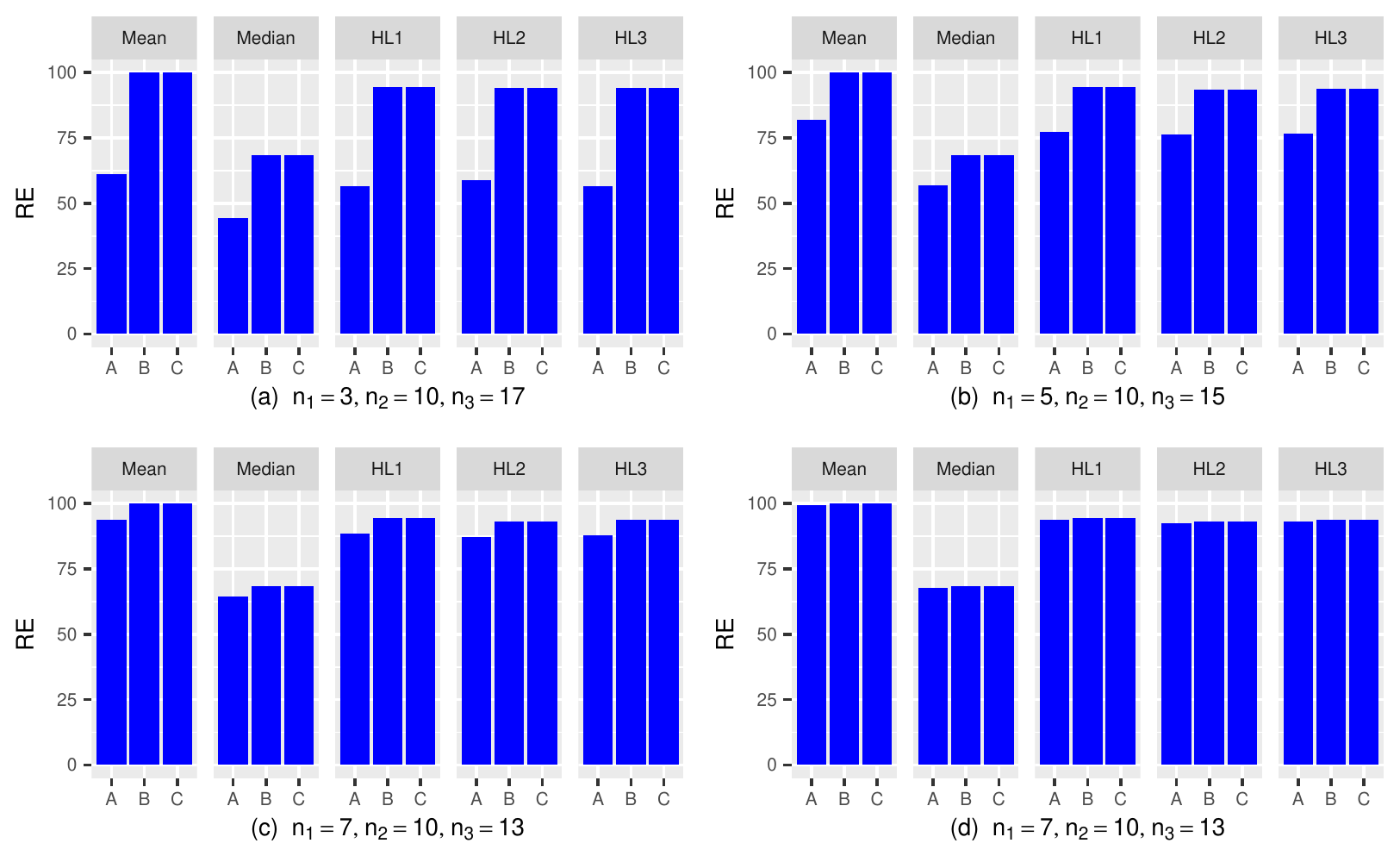}  
\caption{Relative efficiency of the location estimators 
under consideration with different sample sizes. \label{FIG:RElocationpure}}
\bigskip
\bigskip
\centering\includegraphics{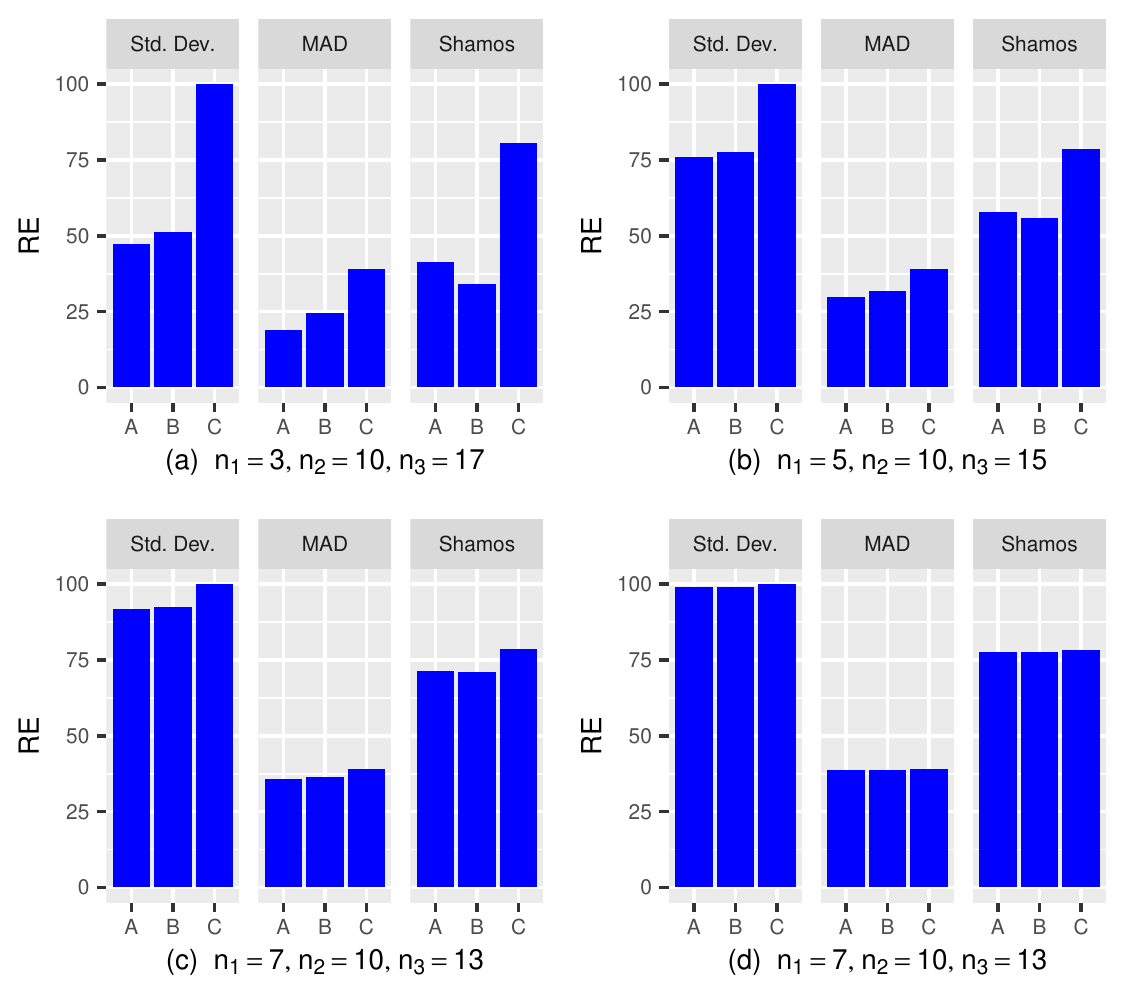}  
\caption{Relative efficiency of the scale estimators
under consideration with different sample sizes. \label{FIG:REscalepure}}
\end{figure}

\subsection{Performance under contamination}
Next, we compare the performance of the proposed methods in (\ref{EQ:muC}) and  (\ref{EQ:sigmaC})
when there is data contamination. Since the estimators tend to be biased in the presence of data contamination, we use the ratio of two mean square errors (MSEs) instead of the ratio of two variances; see, for example, \cite{Park/Leeds:2016} and
\cite{Park/Ouyang/Byun/Leeds:2017,Park/Ouayng/Wang:2021}. Note that using the ratio of the two MSEs allows one to compare the performance of the
estimators based on their variance and bias as well.
Then we have 
$\mathrm{RE}(\hat{\theta}_1 \mid \hat{\theta}_0) = {\mathrm{MSE}(\hat{\theta}_0)}/{\mathrm{MSE}(\hat{\theta}_1)} \times100$ 
(in percentage). By using the estimator with no contamination data as the baseline estimator, we have 
\begin{equation} \label{EQ:REMSE}
\mathrm{RE}(\hat{\theta}_1 \mid \hat{\theta}_0) 
  = \frac{\mathrm{MSE}(\hat{\theta}_0)\textrm{~with no contamination}}%
{\mathrm{MSE}(\hat{\theta}_1)\textrm{~with contamination}} \times 100 .
\end{equation}

We still consider three samples with the four different scenarios with the sample sizes of each scenario given by
(a) $n_1=3$, $n_2=10$, $n_3=17$;
(b) $n_1=5$, $n_2=10$, $n_3=15$;
(c) $n_1=7$, $n_2=10$, $n_3=13$; and
(d) $n_1=9$, $n_2=10$, $n_3=11$.
For the three samples, we generated $X_{ij}$ from the normal distribution 
with mean $\mu_0=100$ and standard deviation $\sigma_0=10$.
It should be noted that there are estimators where a single bad observation can break down estimators. 
For more details, one can refer to Section 11.2 of \cite{Huber/Ronchetti:2009}.
To investigate the impact of a single bad observation, 
we contaminate only the last observation in the second sample by adding $\delta=100$.
Then, for each of the above four scenarios,
we estimate $\mu_0$ and $\sigma_0$ using the three types of estimation methods 
as we did in Section~\ref{SEC:performance}.  
We repeat this simulation one million times ($I=10^6$) to obtain the empirical biases, 
variances, MSE, and RE values summarized 
Table \ref{TBL:varbias1} (scenario a),
Table \ref{TBL:varbias2} (scenario b),
Table \ref{TBL:varbias3} (scenario c), and
Table \ref{TBL:varbias4} (scenario d).
In addition, the MSE values (with squared empirical biases and variances) are plotted in
Figure~\ref{FIG:MSElocation} for the location estimators and
Figure~\ref{FIG:MSEscale} for the scale estimators.
The results clearly show that the conventional estimators (sample mean and standard deviation) are 
very sensitive to data contamination, so that they are seriously biased (even with a single contaminated value), which 
can seriously degrade the performance of the conventional estimators. 
The RE values are summarized in Figures~\ref{FIG:RElocationnoise} and \ref{FIG:REscalenoise}, which also show the underperformance of the conventional estimators.

\begin{table}[thp]
\renewcommand\arraystretch{0.85}
\caption{\label{TBL:varbias1}Estimated variances and biases
($n_1=3$,\quad $n_2=10$,\quad $n_3=17$).}
\begin{center}
\begin{small}
\begin{tabular}{cl*{5}{r}l*{3}{r}}
\hline
 && \multicolumn{5}{c}{Location estimator} && \multicolumn{3}{c}{Scale estimator} \\
 && mean & median & $\mathrm{HL}1$ & $\mathrm{HL}2$ & $\mathrm{HL}3$ && SD & MAD & Shamos \\
\cline{1-1} \cline{3-7} \cline{9-11}
Variance && \\
A        &&  5.4645 & 7.6923 & 6.0507 & 5.8341 & 6.0557 &&  4.6532 & 10.4937 & 5.2532 \\
B        &&  3.3335 & 5.0457 & 3.6957 & 3.7070 & 3.7050 &&  4.3662 &  8.2691 & 6.1166 \\
C        &&  3.3335 & 5.0494 & 3.6951 & 3.7029 & 3.7056 &&  2.5255 &  5.3738 & 2.9547 \\
\cline{1-1}    \cline{3-7} \cline{9-11}
Bias     &&  \\
A        &&  3.3331 & 0.4574 & 0.6964 & 0.6593 & 0.6776 &&  7.9802 &  0.5116 & 1.0234 \\
B        &&  3.3337 & 0.4580 & 0.6970 & 0.6598 & 0.6782 &&  8.1896 &  0.5515 & 0.9536 \\
C        &&  3.3337 & 0.4831 & 0.6962 & 0.6526 & 0.6777 &&  7.9654 &  0.5475 & 0.9904 \\
\cline{1-1}    \cline{3-7} \cline{9-11}
MSE      && \\
A        && 16.5743 & 7.9016 & 6.5357 & 6.2688 & 6.5148 && 68.3374 & 10.7554 & 6.3005 \\
B        && 14.4472 & 5.2555 & 4.1815 & 4.1424 & 4.1659 && 71.4353 &  8.5733 & 7.0258 \\
C        && 14.4472 & 5.2827 & 4.1798 & 4.1289 & 4.1648 && 65.9724 &  5.6735 & 3.9356 \\
\cline{1-1}    \cline{3-7} \cline{9-11}
RE (\%)  && \\
A        && 20.11  &  42.19  & 51.00  & 53.18  & 51.17  &&    2.79 &  17.70  &  30.22 \\
B        && 23.07  &  63.43  & 79.72  & 80.47  & 80.02  &&    2.67 &  22.21  &  27.10 \\
C        && 23.07  &  63.10  & 79.75  & 80.74  & 80.04  &&    2.89 &  33.56  &  48.38 \\
\hline
\end{tabular}
\end{small}
\end{center}
\end{table}

\begin{table}[htbp]
\renewcommand\arraystretch{0.85}
\caption{\label{TBL:varbias2}Estimated variances and biases
($n_1=5$,\quad $n_2=10$,\quad $n_3=15$).}
\begin{center}
\begin{small}
\begin{tabular}{cl*{5}{r}l*{3}{r}}
\hline
 && \multicolumn{5}{c}{Location estimator} && \multicolumn{3}{c}{Scale estimator} \\
 && mean & median & $\mathrm{HL}1$ & $\mathrm{HL}2$ & $\mathrm{HL}3$ && SD & MAD & Shamos \\
\cline{1-1} \cline{3-7} \cline{9-11}
Variance && \\
A        &&  4.0707 & 6.0165 & 4.4794 & 4.5301 & 4.5147  &&  3.1293 & 6.8324 & 3.9380  \\
B        &&  3.3335 & 5.0432 & 3.6917 & 3.7270 & 3.7076  &&  3.0890 & 6.4620 & 4.0363  \\
C        &&  3.3335 & 5.0515 & 3.6911 & 3.7256 & 3.7076  &&  2.5256 & 5.3755 & 3.0535  \\
\cline{1-1}    \cline{3-7} \cline{9-11}
Bias     && \\
A        &&  3.3331 & 0.4599 & 0.6983 & 0.6613 & 0.6798  &&  7.9751 & 0.5136 & 1.0226  \\
B        &&  3.3337 & 0.4603 & 0.6989 & 0.6617 & 0.6802  &&  8.0385 & 0.5243 & 1.0110  \\
C        &&  3.3337 & 0.4857 & 0.6975 & 0.6588 & 0.6802  &&  7.9705 & 0.5499 & 1.0170  \\
\cline{1-1}    \cline{3-7} \cline{9-11}
MSE      && \\
A        && 15.1801 & 6.2280 & 4.9669 & 4.9674 & 4.9769  && 66.7314 & 7.0963 & 4.9836 \\
B        && 14.4472 & 5.2551 & 4.1802 & 4.1648 & 4.1703  && 67.7067 & 6.7369 & 5.0584 \\
C        && 14.4472 & 5.2874 & 4.1775 & 4.1595 & 4.1702  && 66.0547 & 5.6779 & 4.0877 \\
\cline{1-1}    \cline{3-7} \cline{9-11}
RE (\%)  && \\
A        &&  21.96  & 53.52  & 67.11  & 67.11  & 66.98   &&   2.86 &  26.87 &  38.27 \\
B        &&  23.07  & 63.43  & 79.74  & 80.04  & 79.93   &&   2.82 &  28.31 &  37.70\\
C        &&  23.07  & 63.05  & 79.80  & 80.14  & 79.94   &&   2.89 &  33.59 &  46.65\\
\hline
\end{tabular}
\end{small}
\end{center}
\end{table}

\begin{table}[htbp]
\renewcommand\arraystretch{0.85}
\caption{\label{TBL:varbias3}Estimated variances and biases
($n_1=7$,\quad $n_2=10$,\quad $n_3=13$).}
\begin{center}
\begin{small}
\begin{tabular}{cl*{5}{r}l*{3}{r}}
\hline
 && \multicolumn{5}{c}{Location estimator} && \multicolumn{3}{c}{Scale estimator} \\
 && mean & median & $\mathrm{HL}1$ & $\mathrm{HL}2$ & $\mathrm{HL}3$ && SD & MAD & Shamos \\
\cline{1-1} \cline{3-7} \cline{9-11}
Variance &&  \\
A        &&  3.5502 & 5.3307 & 3.9263 & 3.9794 & 3.9536  && 2.6910 &  5.7545 & 3.3054  \\
B        &&  3.3335 & 5.0403 & 3.6926 & 3.7380 & 3.7133  && 2.6835 &  5.6769 & 3.3238  \\
C        &&  3.3335 & 5.0507 & 3.6922 & 3.7377 & 3.7139  && 2.5240 &  5.3844 & 3.0612  \\
\cline{1-1}    \cline{3-7} \cline{9-11}
Bias     && \\
A        &&  3.3341 & 0.4592 & 0.6969 & 0.6598 & 0.6779  && 7.9805 &  0.5090 & 1.0211 \\
B        &&  3.3337 & 0.4587 & 0.6965 & 0.6594 & 0.6776  && 7.9985 &  0.5106 & 1.0186 \\
C        &&  3.3337 & 0.4843 & 0.6957 & 0.6589 & 0.6789  && 7.9796 &  0.5466 & 1.0165 \\
\cline{1-1}    \cline{3-7} \cline{9-11}
MSE      && \\
A        && 14.6663 & 5.5415 & 4.4120 & 4.4148 & 4.4131  && 66.3798 & 6.0137 & 4.3480 \\
B        && 14.4472 & 5.2507 & 4.1778 & 4.1728 & 4.1724  && 66.6592 & 5.9377 & 4.3613 \\
C        && 14.4472 & 5.2853 & 4.1762 & 4.1718 & 4.1748  && 66.1982 & 5.6832 & 4.0943 \\
\cline{1-1}    \cline{3-7} \cline{9-11}
RE (\%)  && \\
A        &&  22.73 &  60.16 & 75.56 & 75.51 & 75.54  &&   2.87 & 31.73  & 43.89 \\
B        &&  23.07 &  63.49 & 79.79 & 79.89 & 79.90  &&   2.86 & 32.14  & 43.75 \\
C        &&  23.07 &  63.07 & 79.82 & 79.91 & 79.85  &&   2.88 & 33.58  & 46.60 \\
\hline
\end{tabular}
\end{small}
\end{center}
\end{table}

\begin{table}[htbp]
\renewcommand\arraystretch{0.85}
\caption{\label{TBL:varbias4}Estimated variances and biases
($n_1=9$,\quad $n_2=10$,\quad $n_3=11$).}
\begin{center}
\begin{small}
\begin{tabular}{cl*{5}{r}l*{3}{r}}
\hline
 && \multicolumn{5}{c}{Location estimator} && \multicolumn{3}{c}{Scale estimator} \\
 && mean & median & $\mathrm{HL}1$ & $\mathrm{HL}2$ & $\mathrm{HL}3$ && SD & MAD & Shamos \\
\cline{1-1} \cline{3-7} \cline{9-11}
Variance && \\
A        &&  3.3552 & 5.0699 & 3.7179 & 3.7616 & 3.7344  && 2.5416 &  5.3737 & 3.0974 \\
B        &&  3.3335 & 5.0408 & 3.6955 & 3.7373 & 3.7110  && 2.5410 &  5.3650 & 3.0989 \\
C        &&  3.3335 & 5.0519 & 3.6952 & 3.7368 & 3.7114  && 2.5248 &  5.3815 & 3.0763 \\
\cline{1-1}    \cline{3-7} \cline{9-11}
Bias     && \\
A        &&  3.3338 & 0.4602 & 0.6979 & 0.6609 & 0.6791  && 7.9790 &  0.5122 & 1.0233 \\
B        &&  3.3337 & 0.4602 & 0.6978 & 0.6608 & 0.6790  && 7.9808 &  0.5106 & 1.0227 \\
C        &&  3.3337 & 0.4858 & 0.6972 & 0.6599 & 0.6799  && 7.9790 &  0.5495 & 1.0230 \\
\cline{1-1}    \cline{3-7} \cline{9-11}
MSE      && \\
A        && 14.4696 & 5.2817 & 4.2050 & 4.1984 & 4.1957  && 66.2064 & 5.6360 & 4.1446  \\
B        && 14.4472 & 5.2526 & 4.1824 & 4.1739 & 4.1720  && 66.2346 & 5.6257 & 4.1448  \\
C        && 14.4472 & 5.2879 & 4.1814 & 4.1723 & 4.1736  && 66.1897 & 5.6834 & 4.1229  \\
\cline{1-1}    \cline{3-7} \cline{9-11}
RE (\%)  && \\
A        && 23.04 &  63.11 & 79.27 & 79.40 & 79.45  &&   2.88 & 33.86 &  46.04\\
B        && 23.07 &  63.46 & 79.70 & 79.87 & 79.90  &&   2.88 & 33.92 &  46.04\\
C        && 23.07 &  63.04 & 79.72 & 79.90 & 79.87  &&   2.88 & 33.58 &  46.29\\
\hline
\end{tabular}
\end{small}
\end{center}
\end{table}

\bigskip
\begin{figure}[thbp]
\centering\includegraphics{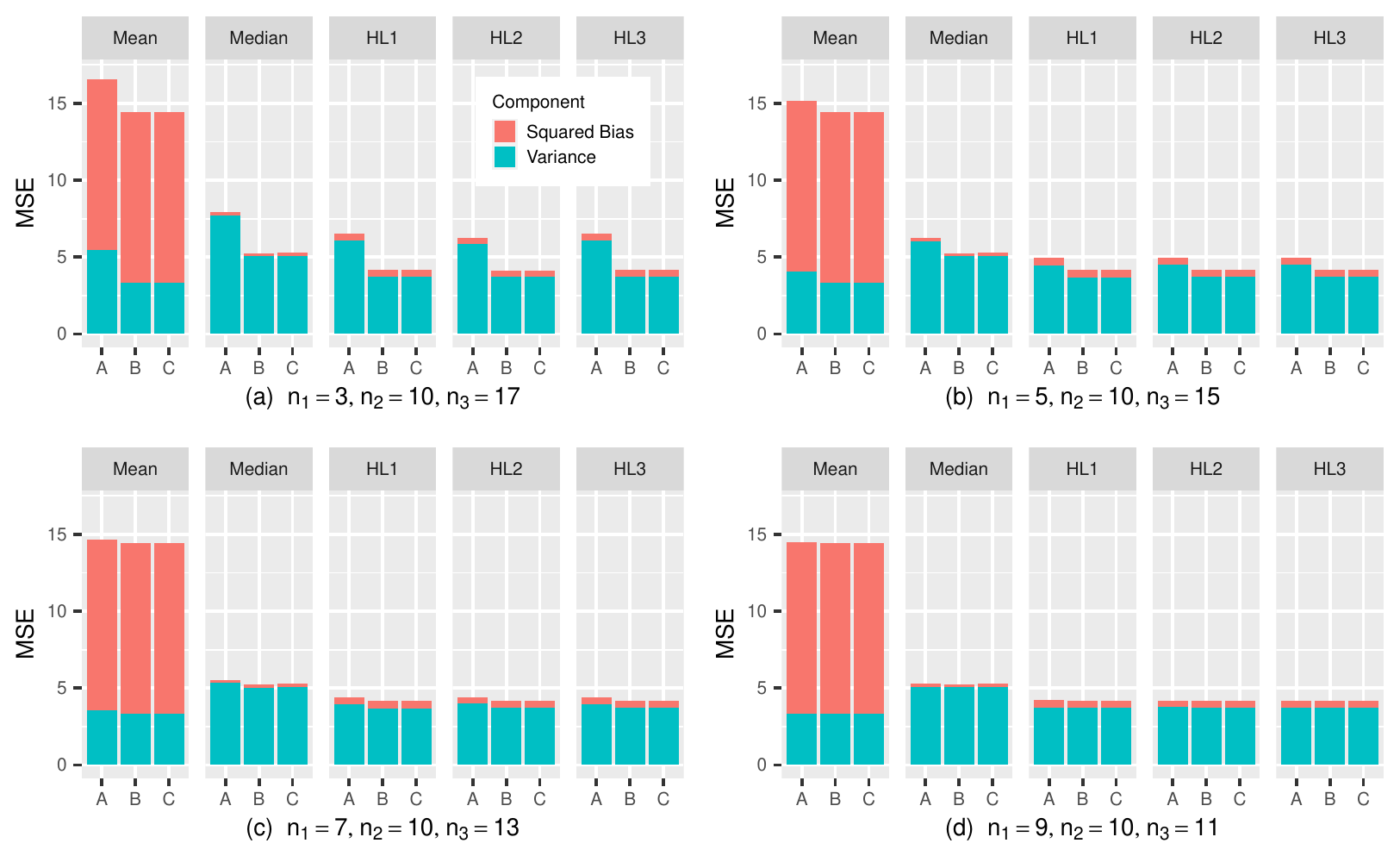}  
\caption{MSE of the location estimators
under consideration with different sample sizes. \label{FIG:MSElocation}}
\bigskip
\bigskip
\centering\includegraphics{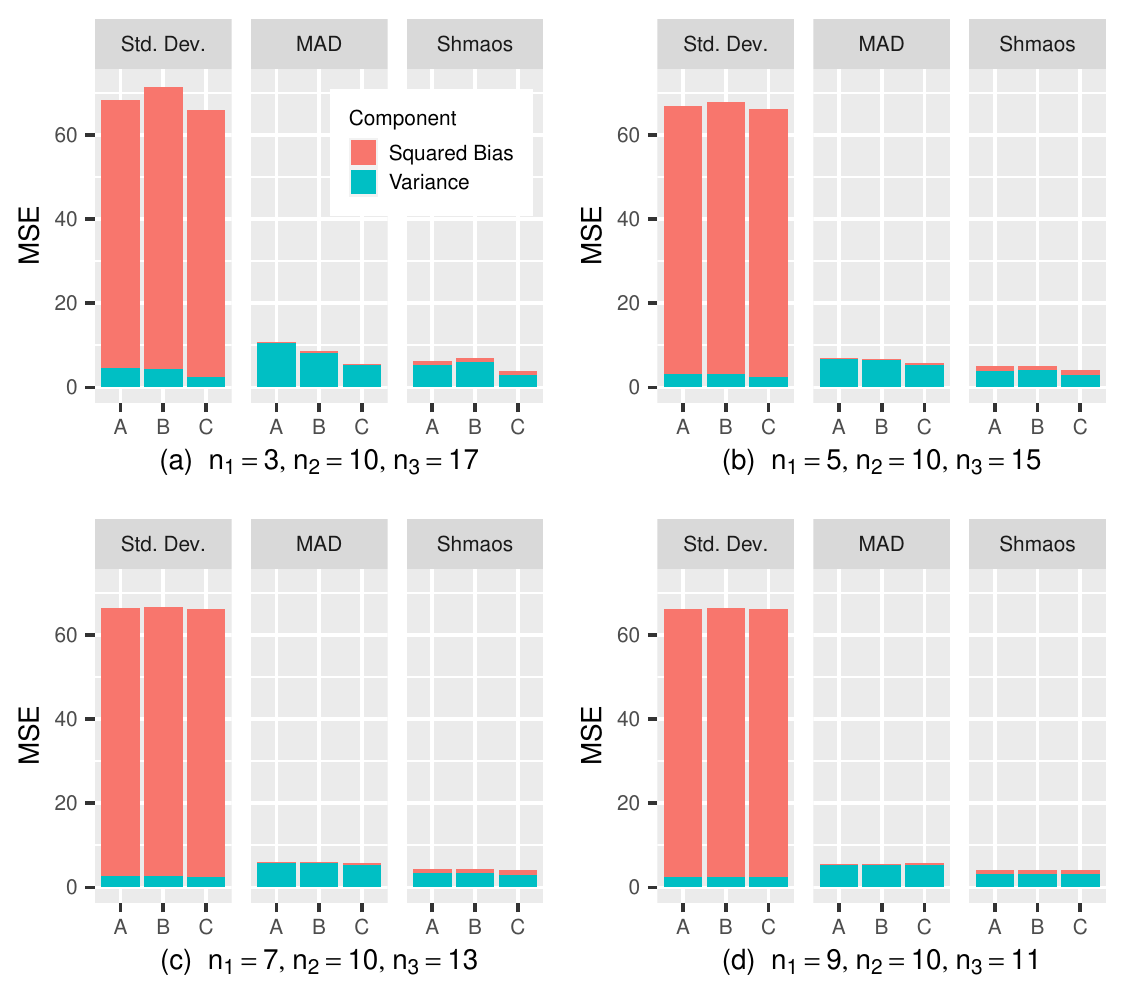}  
\caption{MSE of the scale estimators
under consideration with different sample sizes. \label{FIG:MSEscale}}
\end{figure}

\bigskip
\begin{figure}[thbp]
\centering\includegraphics{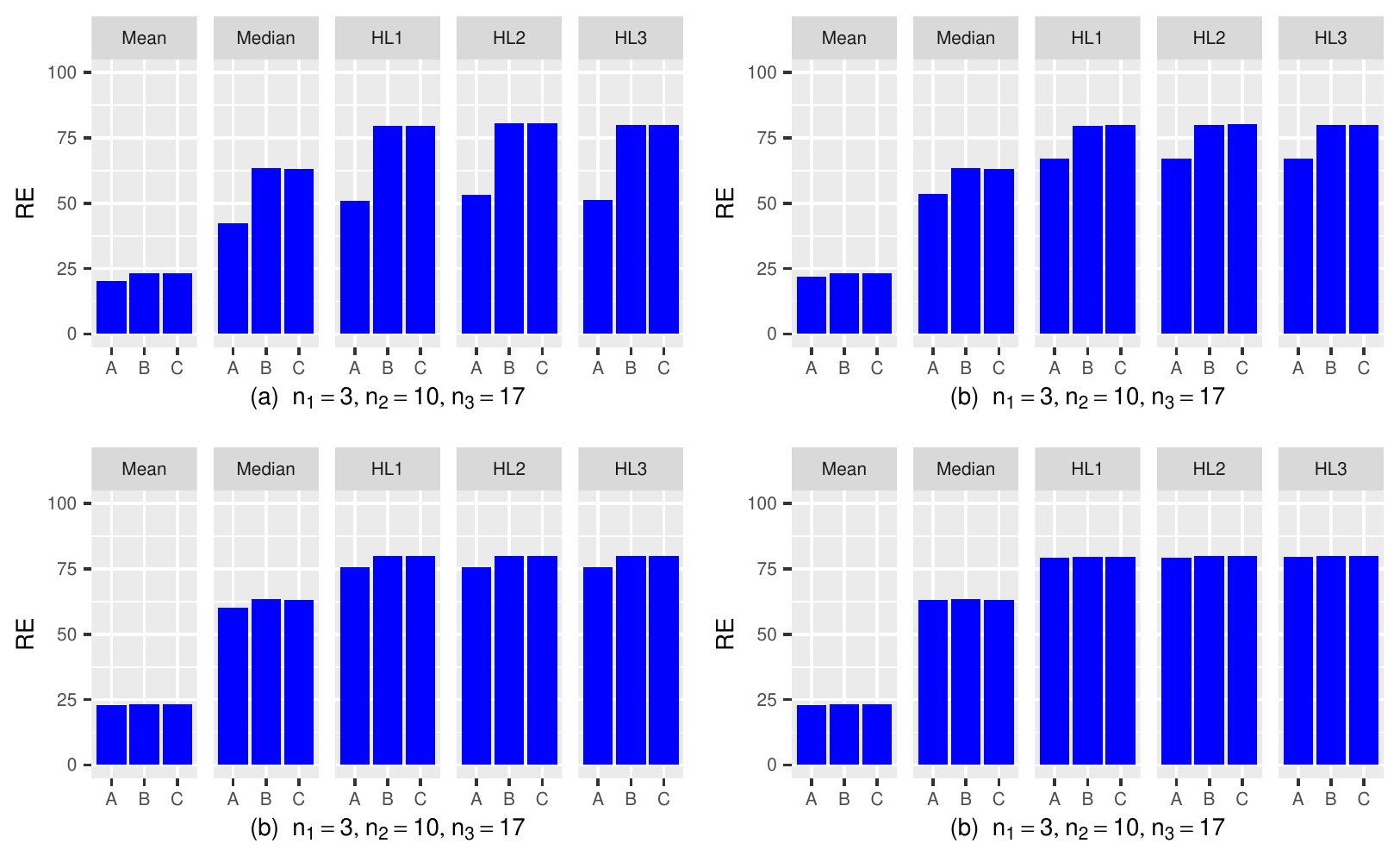}  
\caption{Relative efficiency of the location estimators
under consideration with different sample sizes. 
\label{FIG:RElocationnoise}}
\bigskip
\bigskip
\centering\includegraphics{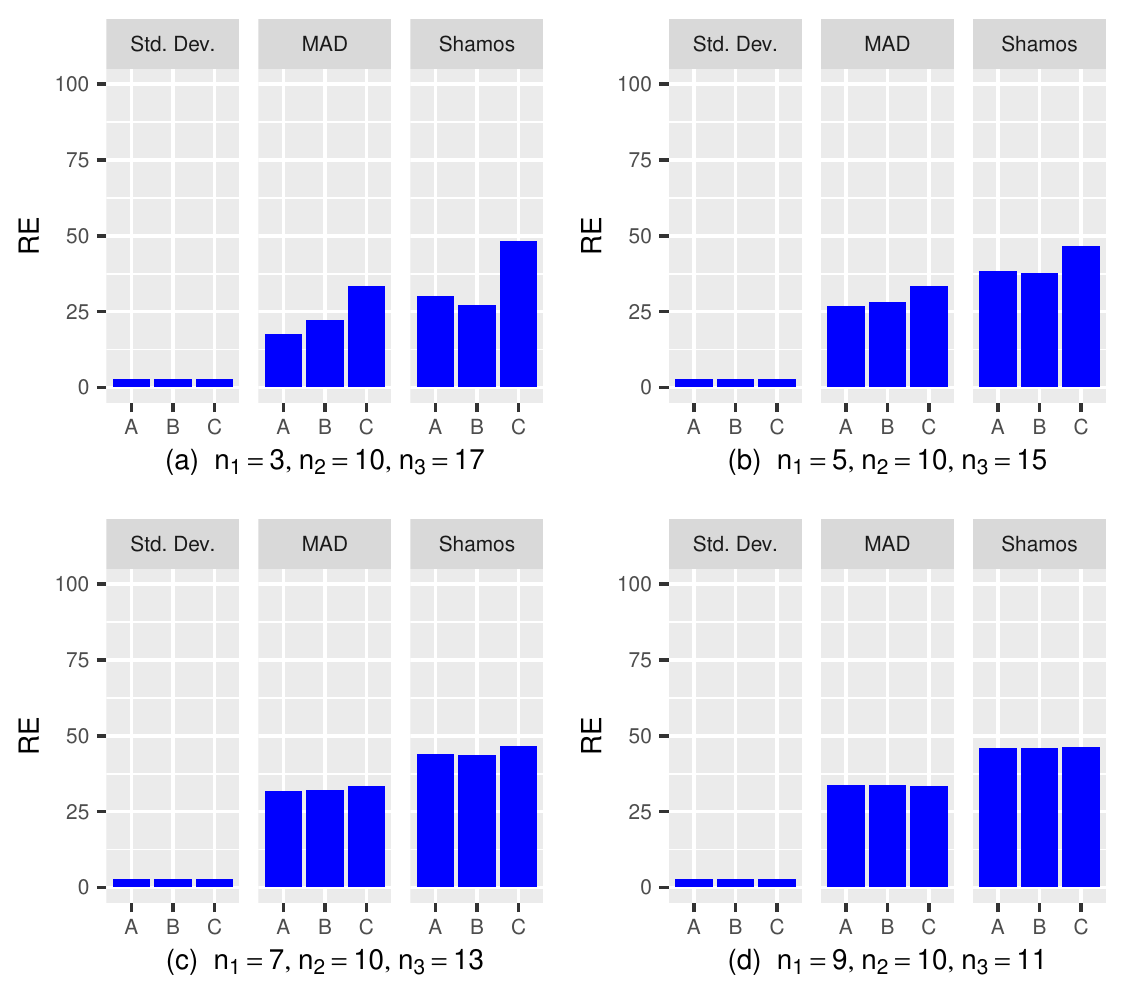}  
\caption{Relative efficiency of the scale estimators
under consideration with different sample sizes. \label{FIG:REscalenoise}}
\end{figure}

\clearpage
\section{Construction of robust $\bar{X}$ control chart with unequal sample sizes} \label{section4}
We briefly introduce how to construct the control charts and then discuss 
how to implement the proposed estimators in (\ref{EQ:muC}) and  (\ref{EQ:sigmaC})
to construct the robust $\bar{X}$ charts with unequal sample sizes. We assume that we have $m$ samples from Phase-I period and that each sample in Phase-I has different sample sizes, denoted by $n_i$ where $i=1,2,\ldots,m$.
Then, we monitor the process with a sample of size $n_k$ in Phase-II.

\subsection{Construction of robust $\bar{X}$ chart\label{SEC:Xbarchart}}
It is well-known that the statistical asymptotic theory has  
\[
\frac{\bar{X}_k - E(\bar{X}_k)}{\mathrm{SE}(\bar{X}_k)} \stackrel{d}{\longrightarrow} N(0,1),
\]
where $\stackrel{d}{\longrightarrow}$ denotes the convergence in distribution
and $\bar{X}_k$ is the sample mean with sample size $n_k$.
Solving $(\bar{X}_k - E(\bar{X}_k))/\mathrm{SE}(\bar{X}_k) = \pm g$ for $\bar{X}_k$, 
we have the $\mathrm{CL} \pm g\cdot \mathrm{SE}$ control limits given by 
\[
E(\bar{X}_k) \pm g\cdot\mathrm{SE}(\bar{X}_k) = \mu \pm g\cdot\frac{\sigma}{\sqrt{n_k}}.
\]
The population parameters, $\mu$ and $\sigma$, are generally unknown  and need to be estimated in Phase-I period. 
As long as we obtain the estimates, $\hat\mu$ and $\hat\sigma$, we have
$\mathrm{UCL} = \hat{\mu} + g{\hat{\sigma}}/{\sqrt{n_k}}$, 
$\mathrm{CL}  = \hat{\mu}$, and 
$\mathrm{LCL} = \hat{\mu} - g{\hat{\sigma}}/{\sqrt{n_k}}$.

To estimate $\mu$, one can use any type of $\bar{\hat{\mu}}_A$ in (\ref{EQ:muA}), 
    $\bar{\hat{\mu}}_B$ in (\ref{EQ:XB}), and $\bar{\hat{\mu}}_C$ in (\ref{EQ:muC}).
For $\sigma$, one can use any type of $\bar{\hat{\sigma}}_A$ in (\ref{EQ:sigmaA}), 
    $\bar{\hat{\sigma}}_B$ in (\ref{EQ:sigmaB}), and $\bar{\hat{\sigma}}_C$ in (\ref{EQ:sigmaC}).
As an illustration, using  $\bar{\hat{\mu}}_A$ and $\bar{\hat{\sigma}}_A$ with the American Standard, 
we have $\bar{X}$ chart as follows:
\begin{align*}
\mathrm{UCL}_A &= \bar{\hat{\mu}}_A + \frac{3\bar{\hat{\sigma}}_A}{\sqrt{n_k}}  \\
 \mathrm{CL}_A &= \bar{\hat{\mu}}_A  \\
\mathrm{LCL}_A &= \bar{\hat{\mu}}_A - \frac{3\bar{\hat{\sigma}}_A}{\sqrt{n_k}} .
\end{align*}

We consider the following combinations for each of the three different estimation methods.
\begin{tabbing} 
Method-I: \quad \= using the sample mean and standard deviation. \\
Method-II:      \> using the median and MAD. \\
Method-III:     \> using the HL     and Shamos.    
\end{tabbing} 
Note that the combination of the sample mean and standard deviation performs the best 
under no contamination although they have a zero breakdown point.
Also, the combination of the median and MAD was chosen because
both estimators have a breakdown point of 50\%.
Both the HL and Shamos estimators have a breakdown point of 29\% and perform well in
both the presence and absence of contamination.

\subsection{Numerical studies}
We carry out extensive Monte Carlo simulations to  compare the performance of 
the proposed $\bar{X}$ charts with respect to the ARL, SDRL and the percentile of 
run length (PRL) of the run lengths (RLs).
In Phase-I, we consider the following five different plans. 
Each plan has fifteen samples and their sample sizes are given as follow. 
\begin{tabbing}
Plan-1:\quad\=$n_1=n_2=\cdots=n_5=3$, \quad\=$n_6=n_7=\cdots=n_{10}=10$,\quad\=$n_{11}=n_{12}=\cdots=n_{15}=17$ \\
Plan-2:     \>$n_1=n_2=\cdots=n_5=5$,      \>$n_6=n_7=\cdots=n_{10}=10$,     \>$n_{11}=n_{12}=\cdots=n_{15}=15$ \\
Plan-3:     \>$n_1=n_2=\cdots=n_5=7$,      \>$n_6=n_7=\cdots=n_{10}=10$,     \>$n_{11}=n_{12}=\cdots=n_{15}=13$ \\
Plan-4:     \>$n_1=n_2=\cdots=n_5=9$,      \>$n_6=n_7=\cdots=n_{10}=10$,     \>$n_{11}=n_{12}=\cdots=n_{15}=11$ \\
Plan-5:     \>$n_1=n_2=\cdots=n_5=10$,     \>$n_6=n_7=\cdots=n_{10}=10$,     \>$n_{11}=n_{12}=\cdots=n_{15}=10$
\end{tabbing}
We generate $X_{ij}$ from the normal distribution with $\mu_0=100$ and $\sigma_0=5$ in Phase-I.
Then we construct control charts using each of the three types: 
$(\bar{\hat{\mu}}_A, \bar{\hat{\sigma}}_A)$ pair, 
$(\bar{\hat{\mu}}_B, \bar{\hat{\sigma}}_B)$ pair, and
$(\bar{\hat{\mu}}_C, \bar{\hat{\sigma}}_C)$ pair
with each of the three different estimation methods. 
Then, in Phase-II, we monitor the process with a sample of size $n_k=10$
from the same normal distribution. 
By repeating this experiment $I=100,000$ times, we obtain the $100,000$ run lengths and calculate the ARL, SDRL,  PRL, and skewness of RLs, which are reported in Table~\ref{TBL:ARLXbarpure}. 
We also draw the box-percentile plots with all 99 percentiles shown in Figure~\ref{FIG:RL-Xbar-pure}. 
For brevity, we considered only Plan-1, Plan-3 and Plan-5 for the box-percentile plots.
Considering that the desired ARL0 is 370,
the pooling type C clearly outperforms for all the methods under consideration. 
We can also observe that Method-I and Method-III are quite comparable, while Method-II underperforms.
Also, this can be observed from Figure~\ref{FIG:RL-Xbar-pure}.

To investigate the effect of data contamination, we add a value ($\delta=100$) 
in the last observation of the fifteenth sample
for all five plans in Phase-I. Again, we repeat this experiment $I=100,000$ times and calculate the ARL, SDRL, PRL, and skewness, which are provided in Table~\ref{TBL:ARLXbarnoise} and Figure~\ref{FIG:RL-Xbar-noise}.
Numerical results clearly show that the values using Method~I change dramatically, while
those using Method~II and Method~III are almost kept unchanged. 
This observation clearly shows the superiority of 
the proposed chart based on the process estimators in (\ref{EQ:muC}) and  (\ref{EQ:sigmaC}).

\begin{table}[tp]
\renewcommand\arraystretch{0.85}
\caption{\label{TBL:ARLXbarpure}Estimated ARL, SDRL, PRL (99\%), and skewness of the RLs 
for the $\bar{X}$ charts based on three different methods ($n_k=10$) with no contamination.}
\begin{small}
\begin{tabular}{@{}c*{12}{r}@{}}
\hline
 && \multicolumn{3}{c}{Method-I}
 && \multicolumn{3}{c}{Method-II}
 && \multicolumn{3}{c}{Method-III} \\
 && A & B & C && A & B & C && A & B & C \\
\cline{1-1} \cline{3-5} \cline{7-9} \cline{11-13}
{Plan-1}  &&\multicolumn{11}{l}{($n_1=n_2=\cdots=n_5=3$,~ $n_6=n_7=\cdots=n_{10}=10$,~$n_{11}=n_{12}=\cdots=n_{15}=17$)}\\
ARL       &&   476.1 &  456.5 &  366.8 && 1455.2 &  809.0 &  491.4 &&  519.0 &  599.7 &  382.3  \\
SDRL      &&  1196.7 & 1056.1 &  554.5 &&33077.2 & 5878.8 & 1342.5 && 1506.1 & 2253.9 &  631.0  \\
99\%      &&  4457.1 & 4064.0 & 2539.0 &&17424.1 &10264.1 & 4998.0 && 5199.0 & 6537.1 & 2819.0  \\
skewness  &&    25.9 &  26.4  &    6.5 &&  160.4 &   72.6 &   17.9 &&   25.1 &   41.7 &    7.5  \\
\cline{1-1} \cline{3-5} \cline{7-9} \cline{11-13}

{Plan-2}  &&\multicolumn{11}{l}{($n_1=n_2=\cdots=n_5=5$,~$n_6=n_7=\cdots=n_{10}=10$,~$n_{11}=n_{12}=\cdots=n_{15}=15$)}\\
ARL       &&   392.9 &  391.0 &  367.1 &&  608.8 &  575.9 &  486.8 &&  427.8 &  434.6 &  384.6 \\
SDRL      &&   650.5 &  641.1 &  547.4 && 2548.2 & 2311.0 & 1351.4 &&  926.8 &  957.7 &  639.0 \\
99\%      &&  2977.0 & 2946.0 & 2517.0 && 6964.0 & 6387.0 & 4964.0 && 3581.0 & 3700.0 & 2885.0 \\
skewness  &&     7.1 &    6.9 &    6.1 &&   42.1 &   46.4 &   27.7 &&   26.7 &   25.1 &    7.8 \\
\cline{1-1} \cline{3-5} \cline{7-9} \cline{11-13}

{Plan-3}  &&\multicolumn{11}{l}{($n_1=n_2=\cdots=n_5=7$,~$n_6=n_7=\cdots=n_{10}=10$,~$n_{11}=n_{12}=\cdots=n_{15}=13$)}\\
ARL       &&   373.9 &  374.4 &  367.8 &&  522.0 &  514.3 &  490.2 &&  396.2 &  397.4 &  383.7  \\
SDRL      &&   589.7 &  589.5 &  557.1 && 1606.7 & 1568.5 & 1482.1 &&  685.2 &  690.7 &  631.3  \\
99\%      &&  2676.0 & 2671.0 & 2577.0 && 5583.0 & 5419.0 & 4904.0 && 3094.0 & 3112.0 & 2865.0  \\
skewness  &&     7.5 &    7.5 &    7.1 &&   28.3 &   29.5 &   35.6 &&    7.9 &    8.0 &    7.2  \\
\cline{1-1} \cline{3-5} \cline{7-9} \cline{11-13}

{Plan-4}  &&\multicolumn{11}{l}{($n_1=n_2=\cdots=n_5=9$,~$n_6=n_7=\cdots=n_{10}=10$,~$n_{11}=n_{12}=\cdots=n_{15}=11$)}\\
ARL       &&   368.5 &  368.5 &  367.4 &&  491.9 &  491.5 &  490.3 &&  384.3 &  384.3 &  383.7  \\
SDRL      &&   556.2 &  556.2 &  553.7 && 1289.7 & 1288.5 & 1333.2 &&  631.0 &  631.3 &  626.5  \\
99\%      &&  2549.0 & 2549.0 & 2534.0 && 5202.1 & 5177.0 & 5122.0 && 2881.0 & 2881.0 & 2866.0  \\
skewness  &&     6.4 &    6.4 &    6.5 &&   14.7 &   14.7 &   18.4 &&    6.7 &    6.7 &    6.6  \\
\cline{1-1} \cline{3-5} \cline{7-9} \cline{11-13}

{Plan-5}  &&\multicolumn{11}{l}{($n_1=n_2=\cdots=n_5=10$,~$n_6=n_7=\cdots=n_{10}=10$,~$n_{11}=n_{12}=\cdots=n_{15}=10$)}\\
ARL       &&   368.7 &  368.7 &  368.7 &&  481.1 &  481.1 &  481.1 &&  385.9 &  385.9 &  385.9 \\
SDRL      &&   556.1 &  556.1 &  556.1 && 1280.2 & 1280.2 & 1280.2 &&  642.2 &  642.2 &  642.2 \\
99\%      &&  2545.0 & 2545.0 & 2545.0 && 4826.0 & 4826.0 & 4826.0 && 2913.0 & 2913.0 & 2913.0 \\
skewness  &&     6.4 &    6.4 &    6.4 &&   21.1 &   21.1 &   21.1 &&    7.8 &    7.8 &    7.8 \\
\hline 
\end{tabular}
\end{small}
\end{table}

\begin{table}[tp]
\renewcommand\arraystretch{0.85}
\caption{\label{TBL:ARLXbarnoise}Estimated ARL, SDRL, PRL (99\%), and skewness of the RLs 
for the $\bar{X}$ charts based on three different methods ($n_k=10$) with contamination.
The last value in the fifteenth sample is contaminated by adding $\delta=100$.
}
\begin{small}
\begin{tabular}{@{}c*{12}{r}@{}}
\hline
 && \multicolumn{3}{c}{Method-I}
 && \multicolumn{3}{c}{Method-II}
 && \multicolumn{3}{c}{Method-III} \\
 && A & B & C && A & B & C && A & B & C \\
\cline{1-1} \cline{3-5} \cline{7-9} \cline{11-13}
{Plan-1}  &&\multicolumn{11}{l}{($n_1=n_2=\cdots=n_5=3$,~ $n_6=n_7=\cdots=n_{10}=10$,~$n_{11}=n_{12}=\cdots=n_{15}=17$)}\\
ARL       &&  6178.6 &  6587.0 &  66089.1 &&  1398.6 &   850.2 &  540.2 &&  579.3 &  657.6 &  466.7 \\
SDRL      && 24535.1 & 24566.0 & 175775.1 && 18574.6 &  5890.5 & 1559.6 && 1713.6 & 2447.5 &  798.0 \\
99\%      && 72414.1 & 74755.1 & 700892.6 && 18548.2 & 11143.3 & 5578.0 && 5960.1 & 7396.1 & 3569.0 \\
skewness  &&    35.2 &    34.6 &     18.2 &&   100.2 &    95.2 &   24.0 &&   22.3 &   41.8 &    7.9 \\
\cline{1-1} \cline{3-5} \cline{7-9} \cline{11-13}

{Plan-2}  &&\multicolumn{11}{l}{($n_1=n_2=\cdots=n_5=5$,~$n_6=n_7=\cdots=n_{10}=10$,~$n_{11}=n_{12}=\cdots=n_{15}=15$)}\\
ARL       &&  6036.8 &  6395.5 &  46438.2 &&   650.3 &   613.6 &  540.8 &&  486.5 &  492.1 &  470.4 \\
SDRL      && 14884.3 & 15719.9 & 121929.7 &&  2682.5 &  2316.1 & 1524.5 && 1012.9 & 1042.0 &  802.5 \\
99\%      && 60810.1 & 63944.0 & 491134.8 &&  7492.1 &  7027.1 & 5557.0 && 4224.1 & 4334.0 & 3596.0 \\
skewness  &&    14.5 &    14.3 &     16.7 &&    40.2 &    41.9 &   21.9 &&   12.2 &   12.0 &    7.5 \\
\cline{1-1} \cline{3-5} \cline{7-9} \cline{11-13}

{Plan-3}  &&\multicolumn{11}{l}{($n_1=n_2=\cdots=n_5=7$,~$n_6=n_7=\cdots=n_{10}=10$,~$n_{11}=n_{12}=\cdots=n_{15}=13$)}\\
ARL       &&  7864.9 &  8146.6 &  31086.0 &&   568.1 &   561.1 &  544.0 &&  458.9 &  459.7 &  474.2 \\
SDRL      && 17590.3 & 18306.0 &  79459.0 &&  1768.6 &  1709.7 & 1570.0 &&  825.8 &  834.9 &  821.4 \\
99\%      && 75670.8 & 79065.0 & 321388.3 &&  6213.0 &  5997.0 & 5640.0 && 3663.0 & 3678.0 & 3687.0 \\
skewness  &&     9.5 &     9.6 &     13.8 &&    22.0 &    22.0 &   20.9 &&    8.3 &    8.6 &    7.6 \\
\cline{1-1} \cline{3-5} \cline{7-9} \cline{11-13}

{Plan-4}  &&\multicolumn{11}{l}{($n_1=n_2=\cdots=n_5=9$,~$n_6=n_7=\cdots=n_{10}=10$,~$n_{11}=n_{12}=\cdots=n_{15}=11$)}\\
ARL       && 11724.3 &  11892.9 & 19452.7 &&   545.5 &   546.2 &  546.5 &&  463.7 &  463.2 &  473.9 \\
SDRL      && 25690.7 &  26153.4 & 44571.0 &&  1656.1 &  1655.0 & 1597.2 &&  788.9 &  788.1 &  813.3 \\
99\%      &&115429.0 & 118007.4 &192744.8 &&  5592.0 &  5618.1 & 5639.1 && 3585.0 & 3579.0 & 3679.1 \\
skewness  &&     8.8 &      8.7 &    10.1 &&    33.1 &    33.1 &   31.0 &&    7.2 &    7.2 &    7.4 \\
\cline{1-1} \cline{3-5} \cline{7-9} \cline{11-13}

{Plan-5}  &&\multicolumn{11}{l}{($n_1=n_2=\cdots=n_5=10$,~$n_6=n_7=\cdots=n_{10}=10$,~$n_{11}=n_{12}=\cdots=n_{15}=10$)}\\
ARL       && 15165.5 &  15165.5 & 15165.5 &&   528.5 &   528.5 &  528.5 &&  473.9 &  473.9 &  473.9 \\
SDRL      && 33868.8 &  33868.8 & 33868.8 &&  1323.0 &  1323.0 & 1323.0 &&  800.0 &  800.0 &  800.0 \\
99\%      &&149699.4 & 149699.4 &149699.4 &&  5297.0 &  5297.0 & 5297.0 && 3676.0 & 3676.0 & 3676.0 \\
skewness  &&     8.7 &      8.7 &     8.7 &&    13.0 &    13.0 &   13.0 &&    7.0 &    7.0 &    7.0 \\
\hline 
\end{tabular}
\end{small}
\end{table}

\bigskip
\begin{figure}[ptbh]
\centering\includegraphics{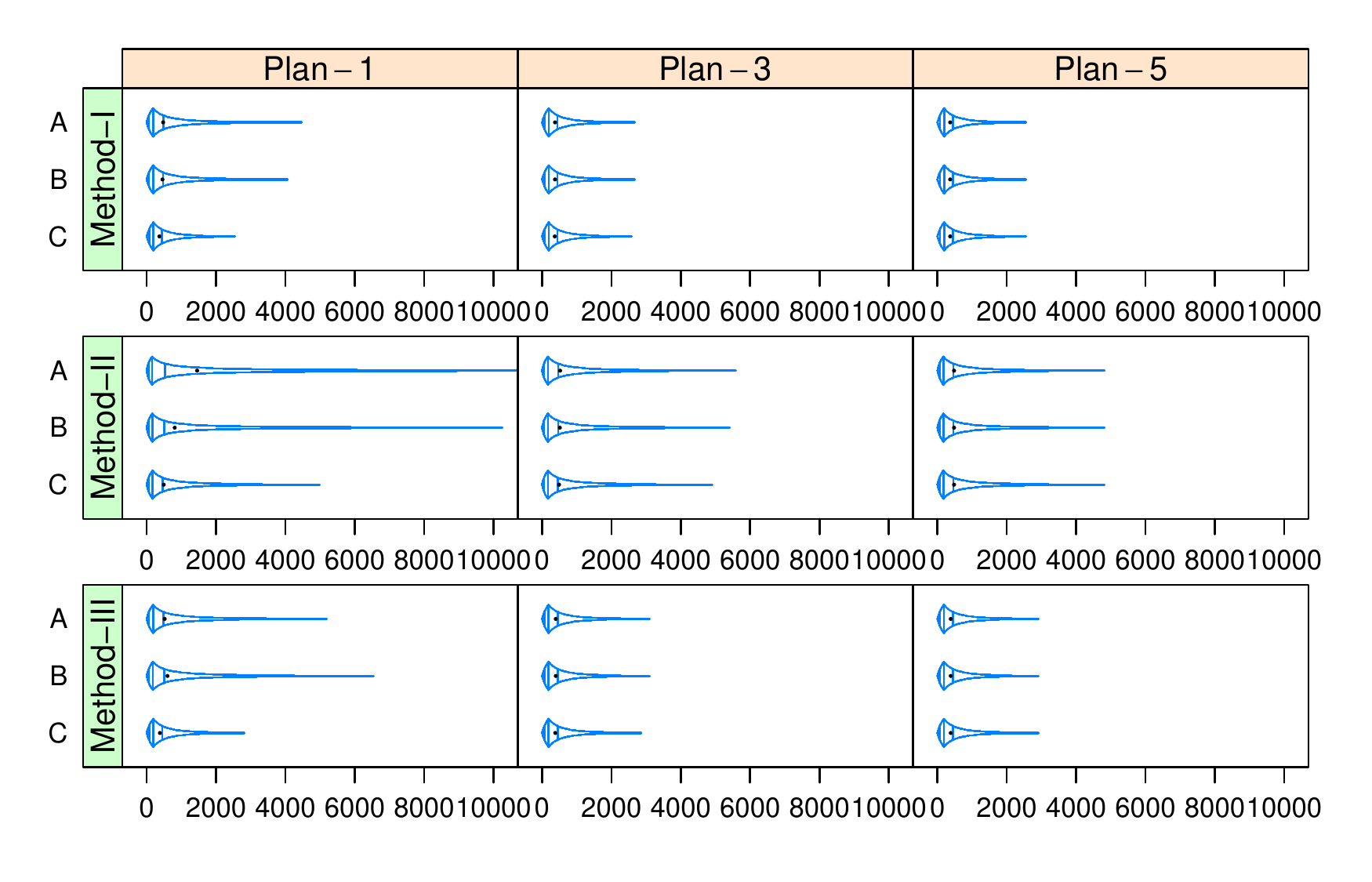}  

\vspace*{-5ex}
\caption{The box-percentile plots of the run lengths of $\bar{X}$ charts under consideration 
(with no contamination).\label{FIG:RL-Xbar-pure}}

\centering\includegraphics{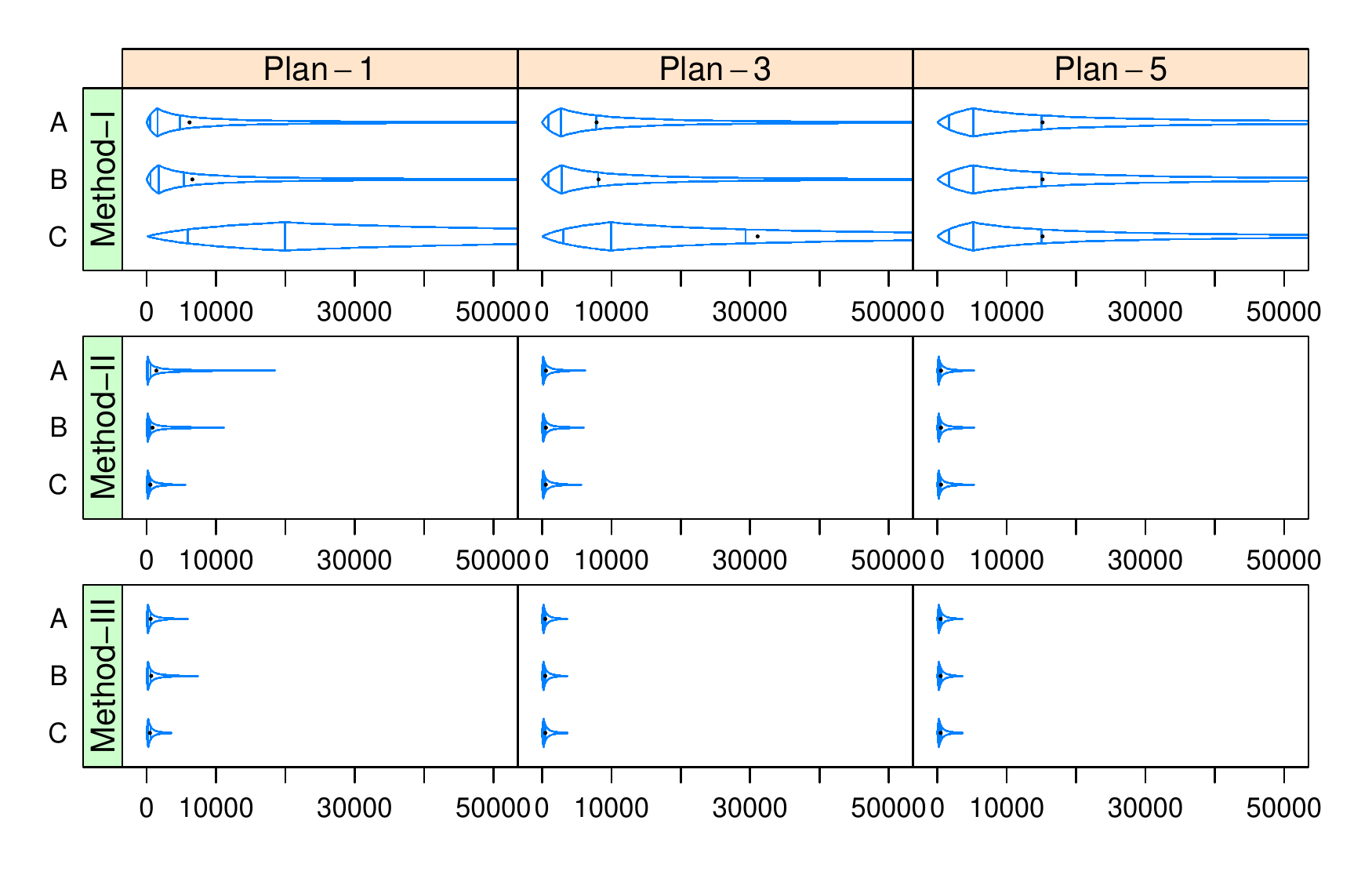}  

\vspace*{-5ex}
\caption{The box-percentile plots of the run lengths of $\bar{X}$ charts under consideration
(with contamination).\label{FIG:RL-Xbar-noise}}
\end{figure}

 \clearpage
\section{An illustrative example} \label{section5}
In this section, we consider a real-data set from Table~6.4 of \cite{Montgomery:2013a}. The data set includes 113 measurements (in millimeters) of the diameters of piston rings
for an automotive engine produced by a forging process.
The observations were obtained from twenty-five samples and the sample sizes are
given by 5, 3, 5, 5, 5, 4, 4, 5, 4, 5, 5, 5, 3, 5, 3, 5, 4, 5, 5, 3, 5, 5, 5, 5, 5.
To construct the $\bar{X}$ chart with the size $n_k=5$, we used the pooling type C in (\ref{EQ:muC}) and (\ref{EQ:sigmaC}) 
for all the methods (Methods-I, II, III mentioned in Subsection~\ref{SEC:Xbarchart}). 
It should be noted that as an estimate of $\sigma$, \cite{Montgomery:2013a} used the square root of the pooled sample
variance, which is different from the pooled standard deviation with the pooling type C.

\begin{table}[htbp]
\caption{\label{TBL:Example-mw4-1}
Control limits of the $\bar{X}$ charts with $n_k=5$.}
\centering
\begin{tabular}{clccc} 
\hline
\cline{1-1} \cline{3-5}
Montgomery    && 73.98606 & 74.00075 & 74.01544  \\
Method-I   && 73.98693 & 74.00075 & 74.01457 \\
Method-II  && 73.98650 & 74.00139 & 74.01629 \\
Method-III && 73.98644 & 74.00072 & 74.01499 \\
\hline
\end{tabular}
\end{table}

In this example, we investigate the sensitivity of the considered methods due to  data contamination. 
To be more specific, 
we add a single contaminated observation (denoted by $\delta$) in the first sample and then calculate 
LCL, CL, and UCL. 
To further investigate their effect due to the level of contaminated value, we changed the value of $\delta$ in a grid-like fashion such as $\delta=73(0.1)74$. Then we consider an empirical approach analogous to the underlying spirit of the influence function by investigating how any method is affected by the level of contaminated value. For the illustrative examples on the implementation of this empirical procedure, one may refer to Figure~5 of \cite{Park/Basu:2011} and Figure~3 of \cite{Park/Kim/Wang:2022}.
Numerical results are depicted in Figure~\ref{FIG:Example-mw4-1}. We observe that the $\bar{X}$ control charts based on \cite{Montgomery:2013a} and Method-I are 
very sensitive to contamination, whereas those based on Method-II and Method-III are quite robust.

\bigskip
\begin{figure}[tph]
\centering\includegraphics{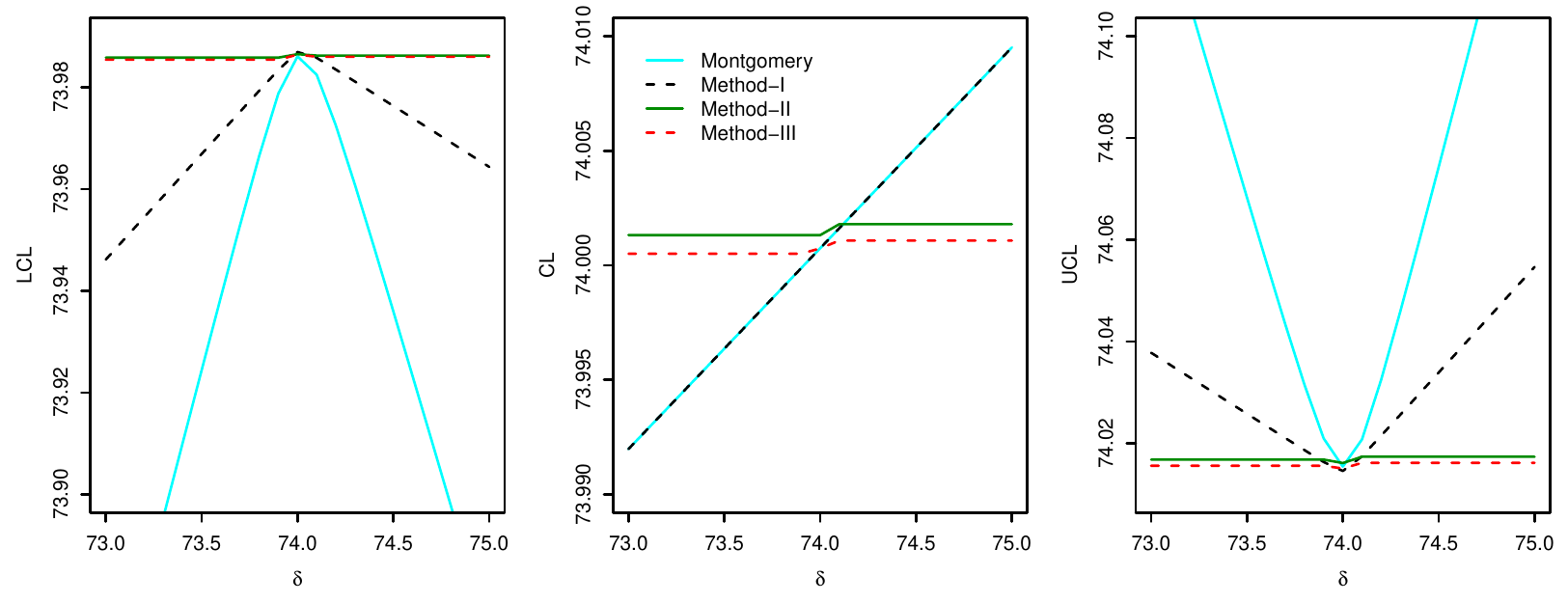}  
\caption{LCL, CL and UCL of $\bar{X}$ chart, and their effect due to the amount of contamination
level ($\delta$).\label{FIG:Example-mw4-1}}
\end{figure}

\section{Concluding remarks} \label{section6}
In this paper, we have developed robust $\bar{X}$ charts, which can simultaneously deal with the issues of unequal sample sizes and data contamination. We showed that in the construction of an $\bar{X}$ control chart with unequal sample sizes, special care should be taken, because inappropriate pooling of estimates from different sample sizes can lead to misleading results. 
To overcome this challenge, 
we have proposed the two process parameter estimation methods in (\ref{EQ:muC}) (location) 
and (\ref{EQ:sigmaC}) (scale)
for optimal pooling in the sense that they are the BLUEs for the unknown process parameters. 
In addition, we have investigated the effect of data contamination on the performance of  robust $\bar{X}$ charts with these BLUEs in terms of 
the ARL, SDRL, PRL and skewness. 
Numerical results from extensive Monte Carlo simulations and a real data analysis revealed that the traditional variable control charts seriously underperform for monitoring process in the presence of data contamination and are extremely sensitive to even a single 
contaminated value, while the proposed robust $\bar{X}$ control charts outperform 
in a manner that is comparable with existing ones, 
whereas they are far superior when the data are contaminated by outliers.  

It should be noteworthy that the proposed robust $X$-bar charts based on the proposed methods 
in (\ref{EQ:muC}) and  (\ref{EQ:sigmaC}) rely on a key assumption that the underlying distribution 
is normally distributed, thus limiting their practical applications in statistical process monitoring, 
since this distribution is (almost always) unknown. 
A natural way to deal with the non-normality model departure issue is to adopt a certain type of data 
transformations to achieve approximate normality, whereas \cite{Khakifirooz:2021} recently pointed out 
several dangers and pitfalls in the use of nonlinear transformations to obtain the
approximate normality of the process data. 
In ongoing work, we continue to investigate the performance of the robust $X$-bar control charts 
when the underlying distribution exhibits departures from the normal distribution, such as skewed distributions. 

Finally, it is noteworthy that
 we have developed \texttt{rcc()} function in \texttt{rQCC} R package \citep{Park/Wang:2022b}
to construct various robust control charts to help field engineers and practitioners 
to monitor the state of the process.


\bibliographystyle{apalike}
\bibliography{REFmw4}

\begin{thebibliography}{}

\bibitem[Abu-Shawiesh, 2008]{Abu-Shawiesh:2008}
Abu-Shawiesh, M. O.~A. (2008).
\newblock A simple robust control chart based on mad.
\newblock {\em Journal of Mathematics and Statistics}, 4:102--107.

\bibitem[Alloway and Raghavachari, 1991]{Allo:Ragh:1991}
Alloway, Jr., J.~A. and Raghavachari, M. (1991).
\newblock Control chart based on the {H}odges-{L}ehmann estimator.
\newblock {\em Journal of Quality Technology}, 23(4):336--347.

\bibitem[Burr, 1969]{Burr:1969}
Burr, I.~W. (1969).
\newblock Control charts for measurements with varying sample sizes.
\newblock {\em Journal of Quality Technology}, 1:163--167.

\bibitem[Faraz et~al., 2019]{Fara:Sani:Mont:2019}
Faraz, A., Saniga, E., and Montgomery, D. (2019).
\newblock Percentile-based control chart design with an application to
  {S}hewhart $\bar{X}$ and {$S^2$} control charts.
\newblock {\em Quality and Reliability Engineering International},
  35(1):116--126.

\bibitem[Fisher, 1922]{Fisher:1922}
Fisher, R.~A. (1922).
\newblock On the mathematical foundations of theoretical statistics.
\newblock {\em Philosophical Transactions of the Royal Society of London.
  Series A, Containing Papers of a Mathematical or Physical Character},
  222:309--368.

\bibitem[Huber and Ronchetti, 2009]{Huber/Ronchetti:2009}
Huber, P.~J. and Ronchetti, E.~M. (2009).
\newblock {\em Robust Statistics}.
\newblock John Wiley \& Sons, New York, 2nd edition.

\bibitem[Janacek and Meikle, 1997]{Janacek/Meikle:1997}
Janacek, G.~J. and Meikle, S.~E. (1997).
\newblock Control charts based on medians.
\newblock {\em Journal of the Royal Statistical Society: Series D (The
  Statistician)}, 46(1):19--31.

\bibitem[Khakifirooz et~al., 2021]{Khakifirooz:2021}
Khakifirooz, M., Tercero-Gómez, V.~G., and Woodall, W.~H. (2021).
\newblock The role of the normal distribution in statistical process
  monitoring.
\newblock {\em Quality Engineering}, 33(3):497--510.

\bibitem[Kim and Reynolds, 2005]{Kim/Reynolds:2005}
Kim, K. and Reynolds, Jr., M.~R. (2005).
\newblock Multivariate monitoring using an mewma control chart with unequal
  sample sizes.
\newblock {\em Journal of Quality Technology}, 37(4):267--281.

\bibitem[Koszty\'an and Katona, 2018]{Kosz:2018}
Koszty\'an, Z.~T. and Katona, A.~I. (2018).
\newblock Risk-based {X}-bar chart with variable sample size and sampling
  interval.
\newblock {\em Computers \& Industrial Engineering}, 120:308--319.

\bibitem[L\`{e}vy-Leduc et~al., 2011]{Levy/etc:2011}
L\`{e}vy-Leduc, C., Boistard, H., Moulines, E., Taqqu, M.~S., and Reisen, V.~A.
  (2011).
\newblock Large sample behaviour of some well-known robust estimators under
  long-range dependence.
\newblock {\em Statistics}, 45:59--71.

\bibitem[Montgomery, 2013]{Montgomery:2013a}
Montgomery, D.~C. (2013).
\newblock {\em Statistical Quality Control: An Modern Introduction}.
\newblock John Wiley \& Sons, Singapore, 7th edition.

\bibitem[Montgomery, 2019]{Montgomery:2019}
Montgomery, D.~C. (2019).
\newblock {\em Introduction to Statistical Quality Control}.
\newblock John Wiley \& Sons, Singapore, 8th edition.

\bibitem[Pappanastos and Adams, 1996]{Papp:Ben:1996}
Pappanastos, E.~A. and Adams, B.~M. (1996).
\newblock Alternative designs of the {H}odges-{L}ehmann control chart.
\newblock {\em Journal of Quality Technology}, 28(2):213--223.

\bibitem[Park and Basu, 2011]{Park/Basu:2011}
Park, C. and Basu, A. (2011).
\newblock Minimum disparity inference based on tangent disparities.
\newblock {\em International Journal of Information and Management Sciences},
  22:1--25.

\bibitem[Park et~al., 2022]{Park/Kim/Wang:2022}
Park, C., Kim, H., and Wang, M. (2022).
\newblock Investigation of finite-sample properties of robust location and
  scale estimators.
\newblock {\em Communication in Statistics -- Simulation and Computation}, To
  appear.

\bibitem[Park and Leeds, 2016]{Park/Leeds:2016}
Park, C. and Leeds, M. (2016).
\newblock A highly efficient robust design under data contamination.
\newblock {\em Computers \& Industrial Engineering}, 93:131--142.

\bibitem[Park et~al., 2017]{Park/Ouyang/Byun/Leeds:2017}
Park, C., Ouyang, L., Byun, J.-H., and Leeds, M. (2017).
\newblock Robust design under normal model departure.
\newblock {\em Computers \& Industrial Engineering}, 113:206--220.

\bibitem[Park et~al., 2021]{Park/Ouayng/Wang:2021}
Park, C., Ouyang, L., and Wang, M. (2021).
\newblock Robust g-type quality control charts for monitoring nonconformities.
\newblock {\em Computers \& Industrial Engineering}, 162:107765.

\bibitem[Park and Wang, 2020]{Park/Wang:2020a}
Park, C. and Wang, M. (2020).
\newblock A study on the {X}-bar and {S} control charts with unequal sample
  sizes.
\newblock {\em Mathematics}, 8(5):698.

\bibitem[Park and Wang, 2022]{Park/Wang:2022b}
Park, C. and Wang, M. (2022).
\newblock \texttt{rQCC}: Robust quality control chart.
\newblock \url{https://CRAN.R-project.org/package=rQCC}.
\newblock R package version 2.22.5 (published on May 23, 2022).

\bibitem[Rocke, 1989]{Rocke:1989}
Rocke, D.~M. (1989).
\newblock Robust control charts.
\newblock {\em Technometrics}, 31(2):173--184.

\bibitem[Serfling, 2011]{Serfling:2011}
Serfling, R.~J. (2011).
\newblock Asymptotic relative efficiency in estimation.
\newblock In Lovric, M., editor, {\em Encyclopedia of Statistical Science, Part
  I}, pages 68--82. Springer-Verlag, Berlin.

\bibitem[Shewhart, 1926]{Shewhart:1926b}
Shewhart, W.~A. (1926).
\newblock Quality control charts.
\newblock {\em Bell Systems Technical Journal}, 5:593--603.

\bibitem[Vining, 2009]{Vining:2009}
Vining, G. (2009).
\newblock Technical advice: {Phase I} and {Phase II} control charts.
\newblock {\em Quality Engineering}, 21(4):478--479.

\bibitem[Wheeler, 2010]{Wheel:2010}
Wheeler, D.~J. (2010).
\newblock Are you sure we don’t need normally distributed data? {M}ore about
  the misuses of probability theory.
\newblock
  \url{https://www.qualitydigest.com/inside/six-sigmacolumn/are-you-sure-we-don-t-need-normally-distributed-data-110110.html}.
\newblock Quality Digest.

\bibitem[Woodall and Faltin, 2019]{Wood:Falt:2019}
Woodall, W.~H. and Faltin, F.~W. (2019).
\newblock Rethinking control chart design and evaluation.
\newblock {\em Quality Engineering}, 31(4):596--605.

\bibitem[Yao and Chakraborti, 2021]{Yao:Chak:2021}
Yao, Y. and Chakraborti, S. (2021).
\newblock Phase {I} monitoring of individual normal data: Design and
  implementation.
\newblock {\em Quality Engineering}, 33(3):443--456.

\end{thebibliography}

\end{document}